\documentclass[]{eptcs}
\usepackage{enumerate}
\usepackage{hyperref}

\usepackage{amssymb,nicefrac}
\usepackage{amsfonts}
\usepackage{graphicx}
\usepackage{amsmath}
\usepackage[varg]{txfonts}
\usepackage{stmaryrd}

\usepackage{framed}
\usepackage{color}
\usepackage{ulem}
\usepackage{tikzfig}

\usepackage{turnstile}

\tikzstyle{env}=[copoint,regular polygon rotate=0,minimum width=0.2cm, fill=black]

\tikzstyle{every picture}=[baseline=-0.25em]
\tikzstyle{dotpic}=[scale=0.5]
\tikzstyle{diredges}=[every to/.style={diredge}]
\tikzstyle{dot graph}=[shorten <=-0.1mm,shorten >=-0.1mm,scale=0.6]
\tikzstyle{plot point}=[circle,fill=black,minimum width=2mm,inner sep=0]

\tikzstyle{braceedge}=[decorate,decoration={brace,amplitude=2mm,raise=-1mm}]
\tikzstyle{small braceedge}=[decorate,decoration={brace,amplitude=1mm,raise=-1mm}]
\tikzstyle{left hook arrow}=[left hook-latex]
\tikzstyle{right hook arrow}=[right hook-latex]

\tikzstyle{black dot}=[inner sep=0.7mm,minimum width=0pt,minimum height=0pt,fill=black,draw=black,shape=circle]

\tikzstyle{dot}=[black dot]
\tikzstyle{smalldot}=[inner sep=0.4mm,minimum width=0pt,minimum height=0pt,fill=black,draw=black,shape=circle]
\tikzstyle{white dot}=[dot,fill=white]
\tikzstyle{antipode}=[white dot,inner sep=0.3mm,font=\footnotesize]
\tikzstyle{smallwhitedot}=[smalldot,fill=white]
\tikzstyle{alt white dot}=[white dot,label={[xshift=3.07mm,yshift=-0.05mm,font=\footnotesize]left:$*$}]
\tikzstyle{gray dot}=[dot,fill=gray!40!white]
\tikzstyle{smallgraydot}=[smalldot,fill=gray!40!white]
\tikzstyle{box vertex}=[draw=black,rectangle]
\tikzstyle{small box}=[box vertex,fill=white]
\tikzstyle{whitebg}=[fill=white,inner sep=2pt]
\tikzstyle{graph state vertex}=[sg vertex,fill=black]

\tikzstyle{wide copoint}=[fill=white,draw=black,shape=isosceles triangle,shape border rotate=90,isosceles triangle stretches=true,inner sep=1pt,minimum width=1.5cm,minimum height=5mm]
\tikzstyle{wide point}=[fill=white,draw=black,shape=isosceles triangle,shape border rotate=-90,isosceles triangle stretches=true,inner sep=1pt,minimum width=1.5cm,minimum height=4mm]
\tikzstyle{very wide copoint}=[fill=white,draw=black,shape=isosceles triangle,shape border rotate=-90,isosceles triangle stretches=true,inner sep=1pt,minimum width=2.5cm,minimum height=4mm]
\tikzstyle{very wide empty copoint}=[draw=black,shape=isosceles triangle,shape border rotate=-90,isosceles triangle stretches=true,inner sep=1pt,minimum width=2.5cm,minimum height=4mm]
\tikzstyle{symm}=[ultra thick,shorten <=-1mm,shorten >=-1mm]

\tikzstyle{square box}=[rectangle,fill=white,draw=black,minimum height=5mm,minimum width=5mm,font=\small]
\tikzstyle{square gray box}=[rectangle,fill=gray!30,draw=black,minimum height=6mm,minimum width=6mm]
\tikzstyle{copoint}=[regular polygon,regular polygon sides=3,draw=black,scale=0.75,inner sep=-0.5pt,minimum width=7mm,fill=white]
\tikzstyle{point}=[regular polygon,regular polygon sides=3,draw=black,scale=0.75,inner sep=-0.5pt,minimum width=7mm,fill=white,regular polygon rotate=180]
\tikzstyle{gray point}=[point,fill=gray!40!white]
\tikzstyle{gray copoint}=[copoint,fill=gray!40!white]

\newcommand{\edgearrow}{{\arrow[black]{>}}}
\newcommand{\edgetick}{{\arrow[black,scale=0.7,very thick]{|}}}

\tikzstyle{diredge}=[->]
\tikzstyle{rdiredge}=[<-]
\tikzstyle{medium diredge}=[->]

\tikzstyle{short diredge}=[->]
\tikzstyle{halfedge}=[-)]
\tikzstyle{other halfedge}=[(-]
\tikzstyle{freeedge}=[(-)]
\tikzstyle{white edge}=[line width=5pt,white]
\tikzstyle{tick}=[postaction=decorate,decoration={markings, mark=at position 0.5 with \edgetick}]
\tikzstyle{small map edge}=[|-latex, gray!60!blue, shorten <=0.9mm, shorten >=0.5mm]
\tikzstyle{thick dashed edge}=[very thick,dashed,gray!40]
\tikzstyle{map edge}=[|-latex,very thick, gray!40, shorten <=1mm, shorten >=0.5mm]
\tikzstyle{tickedge}=[postaction=decorate,
  decoration={markings, mark=at position 0.5 with \edgetick}]
\tikzstyle{dirtickedge}=[postaction=decorate,
  decoration={markings, mark=at position 0.5 with \edgetick},
  decoration={markings, mark=at position 0.85 with \edgearrow}]
\tikzstyle{dirdoubletickedge}=[postaction=decorate,
  decoration={markings, mark=at position 0.4 with \edgetick},
  decoration={markings, mark=at position 0.6 with \edgetick},
  decoration={markings, mark=at position 0.85 with \edgearrow}]

\makeatletter
\newcommand{\boxshape}[3]{%
\pgfdeclareshape{#1}{
\inheritsavedanchors[from=rectangle] 
\inheritanchorborder[from=rectangle]
\inheritanchor[from=rectangle]{center}
\inheritanchor[from=rectangle]{north}
\inheritanchor[from=rectangle]{south}
\inheritanchor[from=rectangle]{west}
\inheritanchor[from=rectangle]{east}
\backgroundpath{
\southwest \pgf@xa=\pgf@x \pgf@ya=\pgf@y
\northeast \pgf@xb=\pgf@x \pgf@yb=\pgf@y

\@tempdima=#2
\@tempdimb=#3

\pgfpathmoveto{\pgfpoint{\pgf@xa - 5pt + \@tempdima}{\pgf@ya}}
\pgfpathlineto{\pgfpoint{\pgf@xa - 5pt - \@tempdima}{\pgf@yb}}
\pgfpathlineto{\pgfpoint{\pgf@xb + 5pt + \@tempdimb}{\pgf@yb}}
\pgfpathlineto{\pgfpoint{\pgf@xb + 5pt - \@tempdimb}{\pgf@ya}}
\pgfpathlineto{\pgfpoint{\pgf@xa - 5pt + \@tempdima}{\pgf@ya}}
\pgfpathclose
}
}}

\boxshape{NEbox}{0pt}{8pt}
\boxshape{SEbox}{0pt}{-8pt}
\boxshape{NWbox}{8pt}{0pt}
\boxshape{SWbox}{-8pt}{0pt}
\makeatother

\tikzstyle{map}=[draw,shape=NEbox,inner sep=7pt]
\tikzstyle{mapdag}=[draw,shape=SEbox,inner sep=7pt]
\tikzstyle{maptrans}=[draw,shape=SWbox,inner sep=7pt]
\tikzstyle{mapconj}=[draw,shape=NWbox,inner sep=7pt]

\tikzstyle{probs}=[shape=semicircle,fill=gray!40!white,draw=black,shape border rotate=180,minimum width=1.2cm]

\tikzstyle{arrs}=[-latex,font=\small,auto]
\tikzstyle{arrow plain}=[arrs]
\tikzstyle{arrow dashed}=[dashed,arrs]
\tikzstyle{arrow bold}=[very thick,arrs]
\tikzstyle{arrow hide}=[draw=white!0,-]
\tikzstyle{arrow reverse}=[latex-]
\tikzstyle{cdnode}=[]

\tikzstyle{gn}=[dot,fill=lime!50,minimum width=0.2cm,inner sep=0.5pt,font=\footnotesize]
\tikzstyle{rn}=[dot,fill=red!50,inner sep=0.5pt,minimum width=0.2cm,font=\footnotesize]
\tikzstyle{bn}=[dot,fill=blue,minimum width=0.3cm]

\tikzstyle{rc}=[dot,thick,fill=white,draw = red,minimum width=0.3cm,inner sep=0pt]
\tikzstyle{gc}=[dot,thick,fill=white,draw= green,inner sep=0pt,minimum width=0.3cm]
\tikzstyle{bc}=[dot,thick,fill=white,draw= blue,minimum width=0.3cm]

\tikzstyle{label}=[circle,fill=white,minimum width=0.3cm]

\tikzstyle{H box}=[rectangle,fill=yellow,draw=black,xscale=1,yscale=1,font=\small,inner sep=0.75pt]

\tikzstyle{clocklabel}=[dot,fill=yellow,draw=black,font=\tiny,inner sep=0.75pt]

\tikzstyle{rsn}=[circle split,draw,fill=red!50,font=\tiny,inner sep=0.75pt]
\tikzstyle{gsn}=[circle split,draw,fill=lime!50,font=\tiny,inner sep=0.75pt]
\tikzstyle{bsn}=[circle split,draw,fill=blue,font=\tiny,inner sep=0.75pt]

\tikzstyle{rsc}=[circle split,thick,draw= red,draw,fill=white,font=\tiny,inner sep=0.75pt]
\tikzstyle{gsc}=[circle split,thick,draw= green,draw,fill=white,font=\tiny,inner sep=0.75pt]
\tikzstyle{bsc}=[circle split,thick,draw= blue,draw,fill=white,font=\tiny,inner sep=0.75pt]

\tikzstyle{cnot}=[fill=white,shape=circle,inner sep=-1.4pt]
\tikzstyle{wire label}=[font=\tiny, auto]

\newcommand{\bra}[1]{\ensuremath{\left\langle #1 \right|}}
\newcommand{\ket}[1]{\ensuremath{\left|  #1 \right\rangle}}

\tikzstyle{cdiag}=[matrix of math nodes, row sep=3em, column sep=3em, text height=1.5ex, text depth=0.25ex,inner sep=0.5em]
\tikzstyle{arrow above}=[transform canvas={yshift=0.5ex}]
\tikzstyle{arrow below}=[transform canvas={yshift=-0.5ex}]

\usepackage{mathtools}
\newtheorem{Th}{Theorem}[section]
\newtheorem{theorem}[Th]{Theorem}

\newtheorem{lemma}[Th]{Lemma}
\newtheorem{corollary}[Th]{Corollary}
\newtheorem{definition}[Th]{Definition}

\newenvironment{proof}{\textbf{Proof:}}{\hfill$\Box$\newline}

\title{Qutrit ZX-calculus is Complete for Stabilizer Quantum Mechanics}
\author{Quanlong Wang
\institute{Department of Computer Science, University of Oxford, UK}
\email{quanlong.wang@cs.ox.ac.uk}
}

\def\titlerunning{Qutrit ZX-calculus is Complete for Stabilizer Quantum Mechanics}
\def\authorrunning{Quanlong Wang}

\begin{document}

\date{}\maketitle

\begin{abstract}
In this paper, we show that a qutrit version of ZX-calculus, with rules significantly different from that of the qubit version,  is complete for pure qutrit stabilizer quantum mechanics, where  state preparations and measurements are based on the three dimensional  computational basis, and unitary operations are required to be in the generalized Clifford group.  This means that any equation of diagrams that holds true under the standard interpretation in Hilbert spaces  can be derived diagrammatically. In contrast to the qubit case, the situation here is more complicated due to the richer structure of this qutrit ZX-calculus.
\end{abstract}

\section{Introduction}
The theory of quantum information and quantum computation (QIC) is traditionally based on binary logic (qubits). However,  multi-valued logic has been recently proposed for quantum computing with linear ion traps \cite{Muthukrishnan}, cold atoms \cite{Smith}, and entangled photons \cite{Malik}. In particular, metaplectic non-Abelian anyons were shown to be naturally connected to ternary (qutrit) logic in contrast to binary logic in topological quantum computing, based on which there comes the Metaplectic Topological Quantum Computer platform \cite{Shawn}. Furthermore,  qutrit-based computers are in certain sense more space-optimal in comparison with quantum computers based on other dimensions of systems \cite{amk}.

The current theoretical tools for qutrit-based QIC are dominated by quantum circuits. However, the quantum circuit notation has  a major defect: the transformation from one circuit diagram into another under the rules of a set of circuit equations. In contrast, the ZX-calculus introduced by Coecke and Duncan \cite{CoeckeDuncan}  is a more intuitive graphical language. Despite being easy to manipulate and seemingly un-mathematical, it is well formulated within the framework of compacted closed categories, which is an important branch of categorical quantum mechanics (CQM) pioneered in 2004 by Abramsky and Coecke \cite{Coeckesamson}. It is intuitive due to its simple rewriting rules which are represented by equations of diagrams for quantum computing. The ZX-calculus  has been successfully applied to QIC in the fields of   (topological) measurement-based quantum computing  \cite{Duncanpx,  Horsman}  and quantum error correction \cite{DuncanLucas,  ckzh}. In particular, the ZX-calculus is pretty useful for qubit stabilizer quantum mechanics because of its completeness for this sub-theory, i.e. any equality that can be derived using matrices can also be derived diagrammatically.

Taking into consideration the practicality of qutrits and the benefits of the ZX-calculus as mentioned above, it is natural to give a qutrit version of ZX-calculus. However, the generalisation from qubits to qutrits is not trivial, since the qutrit-based structures are usually much more complicated then the qubit-based structures. For instance,  the local Clifford group for qubits has only 24 elements, while the local Clifford group for qutrits has 216 elements. Thus it is no surprise that, as presented in   \cite{GongWang},  the rules of qutrit ZX-calculus are significantly different from that of the qubit case: each phase gate has a pair of phase angles, the operator commutation rule is more complicated, the Hadamard gate is not self-adjoint, the colour-change rule is doubled, and the dualiser has a much richer structure than being just an identity.  
Despite being already well established in \cite{BianWang1, BianWang2} and independently introduced as a typical special case of qudit ZX-calculus in \cite{Ranchin}, to the best of our knowledge, there are no completeness results available for qutrit ZX-calculus. Without this kind of results, how can we even know that the rules of  a so-called  ZX-calculus are useful enough for quantum computing?

In this paper, based on the rules and results in \cite{GongWang}, we show that  the qutrit ZX-calculus is complete for pure qutrit stabilizer quantum mechanics (QM). The strategy we used here mirrors from that of the qubit case in \cite{Miriam1}, although it is technically more complicated, especially for the  completeness of the single qutrit Clifford group $\mathcal{C}_1$. Firstly, we show that  any stabilizer state diagram is equal to some GS-LC diagram within the ZX-calculus, 
where a GS-LC diagram consists of a graph state diagram with arbitrary single-qutrit Clifford operators applied to each output. We then show that any stabilizer state diagram can be further brought into a reduced form of the GS-LC diagram. Finally, for any two  stabilizer state diagrams on the same number of qutrits, we make them into a simplified pair of reduced GS-LC diagram such that they are equal under the standard interpretation in Hilbert spaces  if and only if they are identical in the sense that they are constructed by the same element constituents in the same way. By the map-state duality, the case for operators represented by diagrams are also covered, thus we have shown the completeness of the ZX-calculus for  all pure qutrit stabilizer QM.

A natural question will arise at the end of this paper: is there a general proof of completeness of the ZX-calculus for arbitrary dimensional (qudit) stabilizer QM?  This is the problem we would like to address next, but we should also mention some challenges we may face. With exception of the increase of the order of local Clifford groups, the main difficulty comes from the fact that that it is not known whether any stabilizer state is equivalent to a graph state under local Clifford group for the dimension $d$ having multiple prime factors \cite{Hostens}. As far as we know, it is true in the case that $d$ has only single prime factors  \cite{ShiangGr}.  Furthermore, the sufficient and necessary condition for two qudit graph states to be  equivalent under local Clifford group in terms of operations on graphs is unknown in the case that $d$ is non-prime. Finally, the generalised Euler decomposition rule for the generalised Hadamard gate can not be trivially derived, and other uncommon rules might be needed for general $d$.

\section{Qutrit Stabilizer quantum mechanics}

\subsection{ The generalized Pauli group and Clifford group }
 The notions of  Pauli group and Clifford group for qubits can be generalised to 
 qutrits in an natural way:
In the 3-dimensional Hilbert space $H_3$, we define the \textit{ generalised Pauli operators}  $X$ and $Z$ as follows
\begin{equation}
X\ket{j}=\ket{j+1}, ~~Z\ket{j}=\omega^j\ket{j},
\end{equation}
where $ j \in \mathbb{Z}_3 $, $\omega=e^{i\frac{2}{3}\pi}$, and the
 addition is a modulo 3 operation. We will use the same denotation for
 tensor products of these operators as is presented in \cite{Hostens}: for
$v, w \in \mathbb{Z}_3^n$,
$a:=\left(
\begin{array}{c}
v \\
w
\end{array}
\right)\in \mathbb{Z}_3^{2n},$

\begin{equation}\label{xza}
XZ(a):= X^{v_1}Z^{w_1}\otimes \cdots \otimes X^{v_n}Z^{w_n}.
\end{equation}

We define the  \textit{generalized Pauli group} $\mathcal{P}_n$ on $n$ qutrits as
$$
\mathcal{P}_n=\{ \omega^{\delta}XZ(a)| a\in \mathbb{Z}^n_3,  \delta\in \mathbb{Z}_3 \}.
$$

The \textit{ generalized  Clifford group} $\mathcal{C}_n$ on n qutrits is defined as the normalizer of $\mathcal{P}_n$:
$$
\mathcal{C}_n=\{ Q| Q\mathcal{P}_n Q^{\dagger}= \mathcal{P}_n \}.
$$
Notably, for $n=1$, $\mathcal{C}_1$ is called the \textit{ generalized local Clifford group}.
Similar to the qubit case, it can be shown that the generalized Clifford group is generated by  the gate $\mathcal{S}=\ket{0}\bra{0}+\ket{1}\bra{1}+ \omega\ket{2}\bra{2}$,
the generalized Hadamard gate $H=\frac{1}{\sqrt{3}}\sum_{k, j=0}^{2}\omega^{kj}\ket{k}\bra{j}$, and the SUM gate $\Lambda=\sum_{i, j=0}^{2}\ket{i, i+j(mod 3)}\bra{ij}$  \cite{Hostens, Shawn}. In particular, the local Clifford group $\mathcal{C}_1$ is generated by the gate $\mathcal{S}$ and the generalized Hadamard gate $H$ \cite{Hostens}, with the group order being $3^3(3^2-1)=216$, up to global phases \cite{Mark}. 

We define the \textit{stabilizer code} as  the non-zero joint eigenspace to the eigenvalue $1$ of a subgroup of the generalized Pauli group $\mathcal{P}_n$ \cite{Mohsen}. A  \textit{stabilizer state} $\ket{\psi}$  is a stabilizer code of dimension 1, which   
 is therefore stabilized by an abelian subgroup of order $3^{n}$ of the Pauli group excluding multiples of the identity other than the identity itself \cite{Hostens}. We
call this subgroup the \textit{stabilizer}  $\mathfrak{S}$ of $\ket{\psi}$.

\subsection{ Graph states }
Graph states are special stabilizer states which are constructed based on undirected graphs without loops. However, it turns out that they are not far from stabilizer states. 

\begin{definition}\cite{Dan}
A  $\mathbb{Z}_3 $-weighted graph is a pair $G = (V,E)$ where $V$ is a set of $n$ vertices and $E$ is a  collection of weighted edges specified by the adjacency matrix $\Gamma$, which is a symmetric $n$ by $n$ matrix with zero diagonal entries, each  matrix element ${\Gamma_{lm}}\in \mathbb{Z}_3$ representing the weight of the edge connecting vertex $l$ with vertex $m$.
\end{definition}

\begin{definition}\label{Graph State}\cite{Griffiths}
Given a   $\mathbb{Z}_3 $-weighted graph $G = (V,E)$ with $n$ vertices and  adjacency matrix $\Gamma$, 
 the corresponding qutrit graph state can be defined as 
$$\ket{G}=\mathcal{U}\ket{+}^{\otimes n},$$ 
where $\ket{+}=\frac{1}{\sqrt{3}}(\ket{0}+\ket{1}+\ket{2})$,
$\mathcal{U}=\prod_{\{l,m\}\in E}(C_{lm})^{\Gamma_{lm}}$,  $C_{lm}=\Sigma_{j=0}^{2}\Sigma_{k=0}^{2}\omega^{jk}\ket{j}\bra{j}_l\otimes \ket{k}\bra{k}_m$, subscripts indicate to which qutrit the operator is applied.
\end{definition}

\begin{lemma}\cite{Dan}
The qutrit graph state $\ket{G}$ is the unique (up to a global phase) joint $+1$ eigenstate of the group generated by the operators
$$
X_v\prod_{u\in V}(Z_{u})^{\Gamma_{uv}} ~~\mbox{for all}~ v \in V.
$$
\end{lemma}

Therefore,  graph states must be stabilizer states.  On the contrary, stabilizer states are equivalent to graph states in the following sense.
\begin{definition}\cite{Miriam1}
 Two $n$-qutrit stabilizer states $\ket{\psi}$ and $\ket{\phi}$ are equivalent with respect to the local Clifford group if there exists $U \in \mathcal{C}_1^{\otimes n}$ such that $\ket{\psi}=U\ket{\phi}$.
 
 \end{definition}
 
\begin{lemma}\cite{Mohsen} \label{stabilizerequivgraphstate}
 Every qutrit stabilizer state is equivalent to a graph state with respect to the local Clifford group.
 
\end{lemma}

 Below we describe some operations on graphs corresponding to graph states. Theses operations will play a central role in the proof of the completeness of ZX-calculus for qutrit stabilizer quantum mechanics.   

\begin{definition}\cite{Mohsen}
 Let $G = (V,E)$ be a $\mathbb{Z}_3 $-weighted graph with $n$ vertices and  adjacency matrix $\Gamma$. For every vertex $v$, and $0\neq b \in \mathbb{Z}_3$, define the operator $\circ_b v$ on the graph as follows: $G\circ_b v$ is the graph on the same vertex set, with adjacency  matrix $I(v, b)\Gamma I(v, b)$, where $I(v, b) = diag(1,1,...,b,...,1)$, $b$ being on the $v$-th entry.  For every vertex $v$ and $a  \in \mathbb{Z}_3$, define the operator $\ast_a  v $ on the graph as follows: $G\ast_a  v $ is the graph on the same vertex set, with adjacency matrix $\Gamma^{'}$, where
  $\Gamma^{'}_{jk}=\Gamma_{jk}+a\Gamma_{vj}\Gamma_{vk}$ for $j\neq k$, and $\Gamma^{'}_{jj}=0$ for all $j$. The operator $\ast_a  v $ is also called the $a$-local complementation at the vertex $v$ \cite{mmp} .
 
 \end{definition}
 
 Now the equivalence of graph states can be described in terms of these operations on graphs.

 \begin{theorem}\cite{Mohsen} \label{graphstateeuivalence}
 Two graph states  $\ket{G}$  and  $\ket{H}$  with adjacency matrices $M$ and $N$ over $\mathbb{Z}_3 $, are equivalent under local Clifford group if and only if there exists a sequence of $\ast$ and $\circ$ operators acting on one of them to obtain the other.
 \end{theorem}

\section{Qutrit ZX-calculus}
\subsection{The ZX-calculus for general pure state qutrit quantum mechanics}
The ZX-calculus is founded on a symmetric monoidal category (SMC) $\mathfrak{C}$. The objects of $\mathfrak{C}$ are natural numbers: $0, 1, 2,  \cdots$; the tensor of 
objects is just addition of numbers: $m \otimes n = m+n$. The morphisms of $\mathfrak{C}$ are diagrams of the qutrit ZX-calculus. A general diagram  $D:k\to l$   with $k$ inputs and $l$ outputs is generated by:
\begin{center}
\begin{tabular}{|r@{~}r@{~}c@{~}c|r@{~}r@{~}c@{~}c|}
\hline
$R_Z^{(n,m)}(\frac{\alpha}{\beta})$&$:$&$n\to m$ & %
\beginpgfgraphicnamed{Qutrits/generator_spider}
\begin{tikzpicture}
	\begin{pgfonlayer}{nodelayer}
		\node [style=none] (0) at (-1, 1) {};
		\node [style=none] (1) at (1, 1) {};
		\node [style=none] (2) at (-0.75, 0.75) {};
		\node [style=none] (3) at (-0.25, 0.75) {};
		\node [style=none] (4) at (0.75, 0.75) {};
		\node [style=none] (5) at (0.25, 0.5) {...};
		\node [style=gsn] (6) at (0, 0) {$\alpha$\nodepart{lower}$\beta$};
		\node [style=none] (7) at (0.25, -0.5) {...};
		\node [style=none] (8) at (-0.75, -0.75) {};
		\node [style=none] (9) at (-0.25, -0.75) {};
		\node [style=none] (10) at (0.75, -0.75) {};
		\node [style=none] (11) at (-1, -1) {};
		\node [style=none] (12) at (1, -1) {};
	\end{pgfonlayer}
	\begin{pgfonlayer}{edgelayer}
		\draw [style=braceedge] (12.center) to node[wire label, inner sep=5 pt]{$m$} (11.center);
		\draw [style=none, bend left=15, looseness=1.00] (8.center) to (6);
		\draw [style=none, bend left=15, looseness=1.00] (9.center) to (6);
		\draw [style=braceedge] (0.center) to node[wire label, inner sep=5 pt]{$n$} (1.center);
		\draw [style=none, bend left=15, looseness=1.00] (6) to (2.center);
		\draw [style=none, bend right=15, looseness=1.00] (10.center) to (6);
		\draw [style=none, bend right=15, looseness=1.00] (6) to (4.center);
		\draw [style=none, bend left=15, looseness=1.00] (6) to (3.center);
	\end{pgfonlayer}
\end{tikzpicture}}
\endpgfgraphicnamed & $R_X^{(n,m)}(\frac{\alpha}{\beta})$&$:$&$ n\to m$& %
\beginpgfgraphicnamed{Qutrits/generator_spider_gray}
\begin{tikzpicture}
	\begin{pgfonlayer}{nodelayer}
		\node [style=none] (0) at (1, -1) {};
		\node [style=none] (1) at (0.25, -0.5) {...};
		\node [style=none] (2) at (1, 1) {};
		\node [style=none] (3) at (0.75, -0.75) {};
		\node [style=none] (4) at (-0.25, -0.75) {};
		\node [style=none] (5) at (-0.25, 0.75) {};
		\node [style=none] (6) at (0.25, 0.5) {...};
		\node [style=none] (7) at (-0.75, -0.75) {};
		\node [style=rsn] (8) at (0, 0) {$\alpha$\nodepart{lower}$\beta$};
		\node [style=none] (9) at (-1, 1) {};
		\node [style=none] (10) at (-0.75, 0.75) {};
		\node [style=none] (11) at (-1, -1) {};
		\node [style=none] (12) at (0.75, 0.75) {};
	\end{pgfonlayer}
	\begin{pgfonlayer}{edgelayer}
		\draw [style=braceedge] (0.center) to node[wire label, inner sep=5 pt]{$m$} (11.center);
		\draw [style=none, bend left=15, looseness=1.00] (7.center) to (8);
		\draw [style=none, bend left=15, looseness=1.00] (4.center) to (8);
		\draw [style=braceedge] (9.center) to node[wire label, inner sep=5 pt]{$n$} (2.center);
		\draw [style=none, bend left=15, looseness=1.00] (8) to (10.center);
		\draw [style=none, bend right=15, looseness=1.00] (3.center) to (8);
		\draw [style=none, bend right=15, looseness=1.00] (8) to (12.center);
		\draw [style=none, bend left=15, looseness=1.00] (8) to (5.center);
	\end{pgfonlayer}
\end{tikzpicture}}
\endpgfgraphicnamed\\
\hline
$H$&$:$&$1\to 1$ &%
\beginpgfgraphicnamed{RGgenerator/RGg_Hada}
\begin{tikzpicture}
	\begin{pgfonlayer}{nodelayer}
		\node [style={H box}] (0) at (0, -0) {$H$};
		\node [style=none] (1) at (0, 0.5) {};
		\node [style=none] (2) at (0, -0.5) {};
	\end{pgfonlayer}
	\begin{pgfonlayer}{edgelayer}
		\draw (1.center) to (0);
		\draw (2.center) to (0);
	\end{pgfonlayer}
\end{tikzpicture}}
\endpgfgraphicnamed
 & $H^{\dagger}$&$:$&$1\to 1$  &%
\beginpgfgraphicnamed{RGgenerator/RGg_Hadad}
\begin{tikzpicture}
	\begin{pgfonlayer}{nodelayer}
		\node [style={H box}] (0) at (0, -0) {$H^\dagger$};
		\node [style=none] (1) at (0, -0.5) {};
		\node [style=none] (2) at (0, 0.5) {};
	\end{pgfonlayer}
	\begin{pgfonlayer}{edgelayer}
		\draw (2.center) to (0);
		\draw (1.center) to (0);
	\end{pgfonlayer}
\end{tikzpicture}}
\endpgfgraphicnamed \\\hline
  $\sigma$&$:$&$ 2\to 2$& %
\beginpgfgraphicnamed{scalars-s/swap}
\begin{tikzpicture}
	\begin{pgfonlayer}{nodelayer}
		\node [style=none] (0) at (0, 0) {};
		\node [style=none] (1) at (0.3, 0.3) {};
		\node [style=none] (2) at (0.3, -0.3) {};
		\node [style=none] (a) at (-0.3, 0.3) {};
		\node [style=none] (b) at (-0.3, -0.3) {};
		\node [style=none] (3) at (0.3, -0.5) {};
		\node [style=none] (4) at (0.3, 0.5) {};
	\end{pgfonlayer}
	\begin{pgfonlayer}{edgelayer}
		\draw[bend left=35] (1.center) to (0.center);
		\draw[bend right=35] (2.center) to (0.center);
		\draw[bend right=35] (a.center) to (0.center);
		\draw[bend left=35] (b.center) to (0.center);
	\end{pgfonlayer}
\end{tikzpicture}}
\endpgfgraphicnamed &$\mathbb I$&$:$&$1\to 1$&%
\beginpgfgraphicnamed{scalars-s/Id}
\begin{tikzpicture}
	\begin{pgfonlayer}{nodelayer}
		\node [style=none] (1) at (0.5, 0.3) {};
		\node [style=none] (2) at (0.5, -0.3) {};
		\node [style=none] (3) at (0.5, -0.5) {};
		\node [style=none] (4) at (0.5, 0.5) {};
	\end{pgfonlayer}
	\begin{pgfonlayer}{edgelayer}
		\draw (1.center) to (2.center);
	\end{pgfonlayer}
\end{tikzpicture}}
\endpgfgraphicnamed \\\hline
   $e $&$:$&$0 \to 0$& %
\beginpgfgraphicnamed{scalars-s/emptysquare-small}
\begin{tikzpicture}
	\begin{pgfonlayer}{nodelayer}
		\node [style=none] (4) at (0.25, 0.25) {};
		\node [style=none] (5) at (-0.25, 0.25) {};
		\node [style=none] (6) at (0.25, -0.25) {};
		\node [style=none] (7) at (-0.25, -0.25) {};
	\end{pgfonlayer}
	\begin{pgfonlayer}{edgelayer}
		\draw [dashed, color=gray] (5.center) to (7.center);
		\draw [dashed, color=gray] (7.center) to (6.center);
		\draw [dashed, color=gray] (6.center) to (4.center);
		\draw [dashed, color=gray] (4.center) to (5.center);
	\end{pgfonlayer}
\end{tikzpicture}}
\endpgfgraphicnamed &&& &  
  \\\hline
\end{tabular}
\end{center}
where $m,n\in \mathbb N$, $\alpha, \beta \in [0,  2\pi)$, and $e$ represents an empty diagram.

The composition of morphisms is  to combine these components in the following two ways: for any two morphisms $D_1:a\to b$ and $D_2: c\to d$, a \textit{ paralell composition} $D_1\otimes D_2 : a+c\to b+d$ is obtained by placing $D_1$ and $D_2$ side-by-side with $D_1$ on the left of $D_2$;
 for any two morphisms $D_1:a\to b$ and $D_2: b\to c$,  a  \textit{ sequential  composition} $D_2\circ D_1 : a\to c$ is obtained by placing $D_1$ above $D_2$, connecting the outputs of $D_1$ to the inputs of $D_2$.

Spiders with all zero phase angles are simply denoted as follows:
$$%
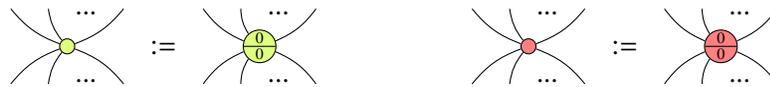
\beginpgfgraphicnamed{scalars-s/spiderg}
\begin{tikzpicture}
	\begin{pgfonlayer}{nodelayer}
		\node [style=none] (0) at (-0.5, 0.5) {};
		\node [style=none] (1) at (-1, 0.5) {};
		\node [style=none] (2) at (0, 0.45) {...};
		\node [style=none] (3) at (0, -0.45) {...};
		\node [style=none] (4) at (0.5, -0.5) {};
		\node [style=none] (5) at (0.75, 0.75) {};
		\node [style=none] (6) at (0.75, -0.75) {};
		\node [style=gn] (7) at (-0.25, -0) {};
		\node [style=none] (8) at (-1.25, 0.75) {};
		\node [style=none] (9) at (-1.25, -0.75) {};
		\node [style=none] (10) at (-1, -0.5) {};
		\node [style=none] (11) at (0.5, 0.5) {};
		\node [style=none] (12) at (-0.5, -0.5) {};
	\end{pgfonlayer}
	\begin{pgfonlayer}{edgelayer}
		\draw [style=none, bend left=15, looseness=1.00] (10.center) to (7);
		\draw [style=none, bend left=15, looseness=1.00] (12.center) to (7);
		\draw [style=none, bend left=15, looseness=1.00] (7) to (1.center);
		\draw [style=none, bend right=15, looseness=1.00] (4.center) to (7);
		\draw [style=none, bend right=15, looseness=1.00] (7) to (11.center);
		\draw [style=none, bend left=15, looseness=1.00] (7) to (0.center);
	\end{pgfonlayer}
\end{tikzpicture}}
\endpgfgraphicnamed := %
\beginpgfgraphicnamed{scalars-s/spidergz}
\begin{tikzpicture}
	\begin{pgfonlayer}{nodelayer}
		\node [style=none] (0) at (-0.5, 0.5) {};
		\node [style=none] (1) at (-1, 0.5) {};
		\node [style=none] (2) at (0, 0.45) {...};
		\node [style=none] (3) at (0, -0.45) {...};
		\node [style=none] (4) at (0.5, -0.5) {};
		\node [style=none] (5) at (0.75, 0.75) {};
		\node [style=none] (6) at (0.75, -0.75) {};
		\node [style=gsn] (7) at (-0.25, -0) {$0$\nodepart{lower}$0$};
		\node [style=none] (8) at (-1.25, 0.75) {};
		\node [style=none] (9) at (-1.25, -0.75) {};
		\node [style=none] (10) at (-1, -0.5) {};
		\node [style=none] (11) at (0.5, 0.5) {};
		\node [style=none] (12) at (-0.5, -0.5) {};
	\end{pgfonlayer}
	\begin{pgfonlayer}{edgelayer}
		\draw [style=none, bend left=15, looseness=1.00] (10.center) to (7);
		\draw [style=none, bend left=15, looseness=1.00] (12.center) to (7);
		\draw [style=none, bend left=15, looseness=1.00] (7) to (1.center);
		\draw [style=none, bend right=15, looseness=1.00] (4.center) to (7);
		\draw [style=none, bend right=15, looseness=1.00] (7) to (11.center);
		\draw [style=none, bend left=15, looseness=1.00] (7) to (0.center);
	\end{pgfonlayer}
\end{tikzpicture}}
\endpgfgraphicnamed\qquad\qquad%
\beginpgfgraphicnamed{scalars-s/spiderr}
\begin{tikzpicture}
	\begin{pgfonlayer}{nodelayer}
		\node [style=none] (0) at (-0.5, 0.5) {};
		\node [style=none] (1) at (-1, 0.5) {};
		\node [style=none] (2) at (0, 0.45) {...};
		\node [style=none] (3) at (0, -0.45) {...};
		\node [style=none] (4) at (0.5, -0.5) {};
		\node [style=none] (5) at (0.75, 0.75) {};
		\node [style=none] (6) at (0.75, -0.75) {};
		\node [style=rn] (7) at (-0.25, -0) {};
		\node [style=none] (8) at (-1.25, 0.75) {};
		\node [style=none] (9) at (-1.25, -0.75) {};
		\node [style=none] (10) at (-1, -0.5) {};
		\node [style=none] (11) at (0.5, 0.5) {};
		\node [style=none] (12) at (-0.5, -0.5) {};
	\end{pgfonlayer}
	\begin{pgfonlayer}{edgelayer}
		\draw [style=none, bend left=15, looseness=1.00] (10.center) to (7);
		\draw [style=none, bend left=15, looseness=1.00] (12.center) to (7);
		\draw [style=none, bend left=15, looseness=1.00] (7) to (1.center);
		\draw [style=none, bend right=15, looseness=1.00] (4.center) to (7);
		\draw [style=none, bend right=15, looseness=1.00] (7) to (11.center);
		\draw [style=none, bend left=15, looseness=1.00] (7) to (0.center);
	\end{pgfonlayer}
\end{tikzpicture}}
\endpgfgraphicnamed := %
\beginpgfgraphicnamed{scalars-s/spiderrz}
\begin{tikzpicture}
	\begin{pgfonlayer}{nodelayer}
		\node [style=none] (0) at (-0.5, 0.5) {};
		\node [style=none] (1) at (-1, 0.5) {};
		\node [style=none] (2) at (0, 0.45) {...};
		\node [style=none] (3) at (0, -0.45) {...};
		\node [style=none] (4) at (0.5, -0.5) {};
		\node [style=none] (5) at (0.75, 0.75) {};
		\node [style=none] (6) at (0.75, -0.75) {};
		\node [style=rsn] (7) at (-0.25, -0) {$0$\nodepart{lower}$0$};
		\node [style=none] (8) at (-1.25, 0.75) {};
		\node [style=none] (9) at (-1.25, -0.75) {};
		\node [style=none] (10) at (-1, -0.5) {};
		\node [style=none] (11) at (0.5, 0.5) {};
		\node [style=none] (12) at (-0.5, -0.5) {};
	\end{pgfonlayer}
	\begin{pgfonlayer}{edgelayer}
		\draw [style=none, bend left=15, looseness=1.00] (10.center) to (7);
		\draw [style=none, bend left=15, looseness=1.00] (12.center) to (7);
		\draw [style=none, bend left=15, looseness=1.00] (7) to (1.center);
		\draw [style=none, bend right=15, looseness=1.00] (4.center) to (7);
		\draw [style=none, bend right=15, looseness=1.00] (7) to (11.center);
		\draw [style=none, bend left=15, looseness=1.00] (7) to (0.center);
	\end{pgfonlayer}
\end{tikzpicture}}
\endpgfgraphicnamed\qquad$$

There are two kinds of rules for the morphisms of $\mathfrak{C}$:  the structure rules for $\mathfrak{C}$ as an SMC, as well as the rewriting rules listed in Figure \ref{figure1}.

Note that all the diagrams should be read from top to bottom.
\begin{figure}[!h]
\begin{center}
\[
\quad \qquad\begin{array}{|cccc|}
\hline
\multicolumn{3}{|c}{%
\beginpgfgraphicnamed{Qutrits/spidernew}
\begin{tikzpicture}[font={\footnotesize}]
	\begin{pgfonlayer}{nodelayer}
		\node [style=none] (0) at (-1.25, -0) {\rotatebox[origin=c]{45}{...}};
		\node [style=none] (1) at (0.25, -0) {$=$};
		\node [style=gsn] (2) at (1.25, -0) {\tiny $\alpha+\eta$\nodepart{lower}\tiny $\beta+\theta$};
		\node [style=gsn] (3) at (-0.75, -0.25) {$\eta$\nodepart{lower}$\theta$};
		\node [style=none] (4) at (-1.75, -0.5) {\raisebox{2mm}{...}};
		\node [style=none] (5) at (1.75, -0.75) {};
		\node [style=none] (6) at (-1, -0.75) {};
		\node [style=none] (7) at (1.25, -0.75) {\raisebox{2mm}{...}};
		\node [style=none] (8) at (-0.5, -0.75) {};
		\node [style=none] (9) at (0.75, -0.75) {};
		\node [style=none] (10) at (-2, -0.5) {};
		\node [style=none] (11) at (-1.5, -0.5) {};
		\node [style=none] (12) at (-0.75, -0.75) {\raisebox{2mm}{...}};
		\node [style=none] (13) at (1.25, 0.75) {\raisebox{-2mm}{...}};
		\node [style=none] (14) at (0.75, 0.75) {};
		\node [style=none] (15) at (-2, 0.75) {};
		\node [style=none] (16) at (-0.5, 0.5) {};
		\node [style=none] (17) at (-1.5, 0.75) {};
		\node [style=none] (18) at (1.75, 0.75) {};
		\node [style=gsn] (19) at (-1.75, 0.25) {$\alpha$\nodepart{lower}$\beta$};
		\node [style=none] (20) at (-0.75, 0.5) {\raisebox{-2mm}{...}};
		\node [style=none] (21) at (-1, 0.5) {};
		\node [style=none] (22) at (-1.75, 0.75) {\raisebox{-2mm}{...}};
		\node [style=none] (23) at (2.75, 0.5) {};
		\node [style=none] (24) at (3.25, 0.5) {};
		\node [style=none] (25) at (4, 0.75) {\raisebox{2mm}{...}};
		\node [style=none] (26) at (2.25, -0) {$=$};
		\node [style=none] (27) at (4.25, -0.5) {};
		\node [style=none] (28) at (2.75, -0.75) {};
		\node [style=none] (29) at (4.25, 0.75) {};
		\node [style=gsn] (30) at (4, 0.25) {$\eta$\nodepart{lower}$\theta$};
		\node [style=none] (31) at (3.75, -0.5) {};
		\node [style=none] (32) at (3.25, -0.75) {};
		\node [style=none] (33) at (4, -0.5) {\raisebox{-2mm}{...}};
		\node [style=none] (34) at (3.75, 0.75) {};
		\node [style=none] (35) at (3, 0.5) {\raisebox{2mm}{...}};
		\node [style=gsn] (36) at (3, -0.25) {$\alpha$\nodepart{lower}$\beta$};
		\node [style=none] (37) at (3.5, -0) {\rotatebox[origin=c]{45}{...}};
		\node [style=none] (38) at (3, -0.75) {\raisebox{-2mm}{...}};
	\end{pgfonlayer}
	\begin{pgfonlayer}{edgelayer}
		\draw (3) to (16.center);
		\draw (3) to (6.center);
		\draw (3) to (8.center);
		\draw (19) to (10.center);
		\draw (19) to (11.center);
		\draw [bend right, looseness=1.00] (19) to (3);
		\draw [bend left, looseness=1.00] (19) to (3);
		\draw (14.center) to (2);
		\draw (2) to (9.center);
		\draw (5.center) to (2);
		\draw (2) to (18.center);
		\draw (19) to (15.center);
		\draw (19) to (17.center);
		\draw (3) to (21.center);
		\draw (30) to (27.center);
		\draw (30) to (34.center);
		\draw (30) to (29.center);
		\draw (36) to (23.center);
		\draw (36) to (24.center);
		\draw [bend left, looseness=1.00] (36) to (30);
		\draw [bend right, looseness=1.00] (36) to (30);
		\draw (36) to (28.center);
		\draw (36) to (32.center);
		\draw (30) to (31.center);
	\end{pgfonlayer}
\end{tikzpicture}}
\endpgfgraphicnamed}&(S1)\\
\beginpgfgraphicnamed{RGrelations/s2}
\begin{tikzpicture}
	\begin{pgfonlayer}{nodelayer}
		\node [style=gn] (0) at (-0.75, -0) {};
		\node [style=none] (1) at (0, -0) {$:=$};
		\node [style=none] (2) at (1, -0.5) {};
		\node [style=none] (3) at (1, 0.5) {};
		\node [style=none] (4) at (-0.75, -0.5) {};
		\node [style=none] (5) at (-0.75, 0.5) {};
		\node [style=gsn] (6) at (1, -0) {\tiny $0$\nodepart{lower}\tiny $0$};
		\node [style=none] (7) at (2.75, 0.5) {};
		\node [style=none] (8) at (2.75, -0.5) {};
		\node [style=none] (9) at (2, -0) {$=$};
	\end{pgfonlayer}
	\begin{pgfonlayer}{edgelayer}
		\draw (6) to (2.center);
		\draw (6) to (3.center);
		\draw (0) to (4.center);
		\draw (0) to (5.center);
		\draw (7.center) to (8.center);
	\end{pgfonlayer}
\end{tikzpicture}}
\endpgfgraphicnamed&(S2)&%
\beginpgfgraphicnamed{Qutrits/cupswap}
\begin{tikzpicture}
	\begin{pgfonlayer}{nodelayer}
		\node [style=none] (0) at (0, -0) {$=$};
		\node [style=none] (1) at (0.5, 0.75) {};
		\node [style=none] (2) at (1.5, -0) {};
		\node [style=none] (3) at (0.5, -0) {};
		\node [style=none] (4) at (1.5, 0.75) {};
		\node [style=none] (5) at (-1.75, 0.5) {};
		\node [style=none] (6) at (-0.25, 0.5) {};
		\node [style=gn] (7) at (-1, -0.25) {};
		\node [style=gn] (8) at (1, -0.5) {};
	\end{pgfonlayer}
	\begin{pgfonlayer}{edgelayer}
		\draw [in=-90, out=-90, looseness=1.50] (3.center) to (2.center);
		\draw [in=90, out=-90, looseness=1.00] (4.center) to (3.center);
		\draw [in=90, out=-90, looseness=1.00] (1.center) to (2.center);
		\draw [bend right=60, looseness=2.00] (5.center) to (6.center);
	\end{pgfonlayer}
\end{tikzpicture}}
\endpgfgraphicnamed&(S3)\\
\beginpgfgraphicnamed{RGrelations/b1}
\begin{tikzpicture}
	\begin{pgfonlayer}{nodelayer}
		\node [style=none] (0) at (0.75, -0.25) {};
		\node [style=none] (1) at (-1, -0.5) {};
		\node [style=none] (2) at (-0.5, -0.5) {};
		\node [style=none] (3) at (0, -0) {$=$};
		\node [style=rn] (4) at (1.25, 0.25) {};
		\node [style=rn] (5) at (-0.75, 0.5) {};
		\node [style=gn] (6) at (-0.75, -0) {};
		\node [style=none] (7) at (1.25, -0.25) {};
		\node [style=rn] (8) at (0.75, 0.25) {};
	\end{pgfonlayer}
	\begin{pgfonlayer}{edgelayer}
		\draw [style=none] (5) to (6);
		\draw [style=none] (6) to (1.center);
		\draw [style=none] (6) to (2.center);
		\draw [style=none] (8) to (0.center);
		\draw [style=none] (4) to (7.center);
	\end{pgfonlayer}
\end{tikzpicture}}
\endpgfgraphicnamed&(B1)&%
\beginpgfgraphicnamed{RGrelations/b2}
\begin{tikzpicture}
	\begin{pgfonlayer}{nodelayer}
		\node [style=none] (0) at (-1.75, 1) {};
		\node [style=rn] (1) at (1.25, 0.25) {};
		\node [style=gn] (2) at (-1, 0.5) {};
		\node [style=none] (3) at (-1.75, -0.75) {};
		\node [style=none] (4) at (1.5, -0.75) {};
		\node [style=none] (5) at (1, 0.75) {};
		\node [style=none] (6) at (0, -0) {$=$};
		\node [style=none] (7) at (1, -0.75) {};
		\node [style=gn] (8) at (-1.75, 0.5) {};
		\node [style=rn] (9) at (-1.75, -0.25) {};
		\node [style=none] (10) at (-1, 1) {};
		\node [style=none] (11) at (-1, -0.75) {};
		\node [style=none] (12) at (1.5, 0.75) {};
		\node [style=gn] (13) at (1.25, -0.25) {};
		\node [style=rn] (14) at (-1, -0.25) {};
	\end{pgfonlayer}
	\begin{pgfonlayer}{edgelayer}
		\draw [style=none] (11.center) to (14);
		\draw [style=none] (3.center) to (9);
		\draw [style=none] (2) to (10.center);
		\draw [style=none, bend right, looseness=1.00] (14) to (2);
		\draw [style=none] (8) to (0.center);
		\draw [style=none, bend left, looseness=1.00] (9) to (8);
		\draw [style=none] (4.center) to (13);
		\draw [style=none] (13) to (1);
		\draw [style=none] (1) to (5.center);
		\draw [style=none] (1) to (12.center);
		\draw (13) to (7.center);
		\draw (8) to (14);
		\draw (2) to (9);
	\end{pgfonlayer}
\end{tikzpicture}}
\endpgfgraphicnamed&(B2)\\
\beginpgfgraphicnamed{Qutrits/k1n}
\begin{tikzpicture}
	\begin{pgfonlayer}{nodelayer}
		\node [style=none] (0) at (-0.5, 1) {};
		\node [style=rn] (1) at (-3.25, -0) {};
		\node [style=gsn] (2) at (0.25, -0) {\tiny $2$\nodepart{lower}\tiny $1$};
		\node [style=rn] (3) at (0.5, 0.5) {};
		\node [style=none] (4) at (-0.75, -0.5) {};
		\node [style=none] (5) at (-2.5, -0.5) {};
		\node [style=none] (6) at (-2.75, 0.25) {$=$};
		\node [style=none] (7) at (-0.25, -0.5) {};
		\node [style=gsn] (8) at (-2.5, -0) {\tiny $1$\nodepart{lower}\tiny $2$};
		\node [style=gsn] (9) at (-0.5, 0.5) {\tiny $2$\nodepart{lower}\tiny $1$};
		\node [style=none] (10) at (0.75, -0.5) {};
		\node [style=none] (11) at (0, 0.25) {$=$};
		\node [style=gsn] (12) at (-2, -0) {\tiny $1$\nodepart{lower}\tiny $2$};
		\node [style=none] (13) at (-2, -0.5) {};
		\node [style=gsn] (14) at (-3.25, 0.5) {\tiny $1$\nodepart{lower}\tiny $2$};
		\node [style=gsn] (15) at (0.75, -0) {\tiny $2$\nodepart{lower}\tiny $1$};
		\node [style=none] (16) at (-3.25, 1) {};
		\node [style=rn] (17) at (-0.5, -0) {};
		\node [style=none] (18) at (-2.25, 1) {};
		\node [style=none] (19) at (-3, -0.5) {};
		\node [style=none] (20) at (0.25, -0.5) {};
		\node [style=rn] (21) at (-2.25, 0.5) {};
		\node [style=none] (22) at (-3.5, -0.5) {};
		\node [style=none] (23) at (0.5, 1) {};
	\end{pgfonlayer}
	\begin{pgfonlayer}{edgelayer}
		\draw [style=none] (1) to (22.center);
		\draw [style=none] (1) to (19.center);
		\draw [style=none] (18.center) to (21);
		\draw [style=none] (17) to (4.center);
		\draw [style=none] (17) to (7.center);
		\draw [style=none] (23.center) to (3);
		\draw [style=none] (16.center) to (14);
		\draw [style=none] (14) to (1);
		\draw [style=none] (21) to (8);
		\draw [style=none] (8) to (5.center);
		\draw [style=none] (21) to (12);
		\draw [style=none] (12) to (13.center);
		\draw [style=none] (0.center) to (9);
		\draw [style=none] (9) to (17);
		\draw [style=none] (3) to (2);
		\draw [style=none] (2) to (20.center);
		\draw [style=none] (3) to (15);
		\draw [style=none] (15) to (10.center);
	\end{pgfonlayer}
\end{tikzpicture}}
\endpgfgraphicnamed&(K1)&%
\beginpgfgraphicnamed{Qutrits/k2n}
\begin{tikzpicture}
	\begin{pgfonlayer}{nodelayer}
		\node [style=gsn] (0) at (-1.5, 0.5) {\tiny $1$\nodepart{lower}\tiny $2$};
		\node [style=rsn] (1) at (-1.5, -0.25) { $\alpha$\nodepart{lower} $\beta$};
		\node [style=none] (2) at (-1, -0) {$=$};
		\node [style=none] (3) at (-1.5, 1) {};
		\node [style=none] (4) at (-1.5, -0.75) {};
		\node [style=none] (5) at (-0.5, -0.75) {};
		\node [style=none] (6) at (-0.5, 1) {};
		\node [style=gsn] (7) at (-0.5, -0.25) {\tiny $1$\nodepart{lower}\tiny $2$};
		\node [style=none] (8) at (0.75, 1) {};
		\node [style=none] (9) at (0.75, -0.75) {};
		\node [style=none] (10) at (1.25, -0) {$=$};
		\node [style=gsn] (11) at (1.75, -0.25) {\tiny $2$\nodepart{lower}\tiny $1$};
		\node [style=none] (12) at (1.75, -0.75) {};
		\node [style=none] (13) at (1.75, 1) {};
		\node [style=rsn] (14) at (0.75, -0.25) { $\alpha$\nodepart{lower} $\beta$};
		\node [style=gsn] (15) at (0.75, 0.5) {\tiny $2$\nodepart{lower}\tiny $1$};
		\node [style=rsn] (16) at (-0.5, 0.5) {\tiny $\beta\textnormal{-}\alpha$\nodepart{lower}\tiny $\textnormal{-}\alpha$};
		\node [style=rsn] (17) at (1.75, 0.5) {\tiny $\textnormal{-}\beta$\nodepart{lower}\tiny $\alpha\textnormal{-}\beta$};
	\end{pgfonlayer}
	\begin{pgfonlayer}{edgelayer}
		\draw (3.center) to (0);
		\draw (0) to (1);
		\draw (1) to (4.center);
		\draw (7) to (5.center);
		\draw (8.center) to (15);
		\draw (15) to (14);
		\draw (14) to (9.center);
		\draw (11) to (12.center);
		\draw (6.center) to (16);
		\draw (16) to (7);
		\draw (13.center) to (17);
		\draw (17) to (11);
	\end{pgfonlayer}
\end{tikzpicture}}
\endpgfgraphicnamed&(K2)\\
%
\beginpgfgraphicnamed{RGrelations/h1}
\begin{tikzpicture}
	\begin{pgfonlayer}{nodelayer}
		\node [style={H box}] (0) at (0, 0.25) {$H^\dagger$};
		\node [style=none] (1) at (-1, 0.75) {};
		\node [style=none] (2) at (1, 0.5) {};
		\node [style=none] (3) at (0, 0.75) {};
		\node [style={H box}] (4) at (-1, -0.25) {$H^\dagger$};
		\node [style=none] (5) at (0.5, -0) {$=$};
		\node [style={H box}] (6) at (-1, 0.25) {$H$};
		\node [style=none] (7) at (0, -0.75) {};
		\node [style=none] (8) at (1, -0.5) {};
		\node [style=none] (9) at (-1, -0.75) {};
		\node [style=none] (10) at (-0.5, -0) {$=$};
		\node [style={H box}] (11) at (0, -0.25) {$H$};
	\end{pgfonlayer}
	\begin{pgfonlayer}{edgelayer}
		\draw (3.center) to (0);
		\draw (7.center) to (11);
		\draw (11) to (0);
		\draw (8.center) to (2.center);
		\draw (1.center) to (6);
		\draw (9.center) to (4);
		\draw (4) to (6);
	\end{pgfonlayer}
\end{tikzpicture}}
\endpgfgraphicnamed&(H1)&%
\beginpgfgraphicnamed{Qutrits/HadaDecom}
\begin{tikzpicture}
	\begin{pgfonlayer}{nodelayer}
		\node [style=none] (0) at (-0.5, 0.5) {};
		\node [style=rsn] (1) at (0.5, -0) {$2$\nodepart{lower}$2$};
		\node [style=gsn] (2) at (0.5, 0.5) {$2$\nodepart{lower}$2$};
		\node [style=none] (3) at (0, -0) {$=$};
		\node [style=none] (4) at (-0.5, -0.5) {};
		\node [style=gsn] (5) at (0.5, -0.5) {$2$\nodepart{lower}$2$};
		\node [style={H box}] (6) at (-0.5, -0) {$H$};
		\node [style=none] (7) at (0.5, -1) {};
		\node [style=none] (8) at (0.5, 1) {};
	\end{pgfonlayer}
	\begin{pgfonlayer}{edgelayer}
		\draw (0.center) to (6);
		\draw (4.center) to (6);
		\draw (8.center) to (2);
		\draw (2) to (1);
		\draw (1) to (5);
		\draw (5) to (7.center);
	\end{pgfonlayer}
\end{tikzpicture}}
\endpgfgraphicnamed&(EU)\\
%
\beginpgfgraphicnamed{RGrelations/h2}
\begin{tikzpicture}
	\begin{pgfonlayer}{nodelayer}
		\node [style=none] (0) at (0.5, 0.75) {};
		\node [style=none] (1) at (-0.5, -1) {};
		\node [style={H box}] (2) at (-0.5, 0.5) {$H$};
		\node [style=none] (3) at (-1, -0.75) {\raisebox{2mm}{...}};
		\node [style=none] (4) at (1, -0.75) {\raisebox{2mm}{...}};
		\node [style=none] (5) at (1, 0.75) {\raisebox{-2mm}{...}};
		\node [style=rsn] (6) at (-1, -0) {$\alpha$\nodepart{lower}$\beta$};
		\node [style={H box}] (7) at (-1.5, -0.5) {$H^\dagger$};
		\node [style=none] (8) at (-0.5, 1) {};
		\node [style={H box}] (9) at (-0.5, -0.5) {$H^\dagger$};
		\node [style=none] (10) at (-1, 0.75) {\raisebox{-2mm}{...}};
		\node [style=none] (11) at (0, -0) {$=$};
		\node [style=gsn] (12) at (1, -0) {$\alpha$\nodepart{lower}$\beta$};
		\node [style=none] (13) at (1.5, -0.75) {};
		\node [style=none] (14) at (1.5, 0.75) {};
		\node [style=none] (15) at (-1.5, 1) {};
		\node [style={H box}] (16) at (-1.5, 0.5) {$H$};
		\node [style=none] (17) at (-1.5, -1) {};
		\node [style=none] (18) at (0.5, -0.75) {};
	\end{pgfonlayer}
	\begin{pgfonlayer}{edgelayer}
		\draw (6) to (7);
		\draw (6) to (9);
		\draw (9) to (1.center);
		\draw (7) to (17.center);
		\draw (6) to (16);
		\draw (6) to (2);
		\draw (2) to (8.center);
		\draw (16) to (15.center);
		\draw (12) to (0.center);
		\draw (12) to (14.center);
		\draw (12) to (18.center);
		\draw (12) to (13.center);
	\end{pgfonlayer}
\end{tikzpicture}}
\endpgfgraphicnamed&(H2)&%
\beginpgfgraphicnamed{Qutrits/h2prime}
\begin{tikzpicture}
	\begin{pgfonlayer}{nodelayer}
		\node [style=none] (0) at (0, -0) {$=$};
		\node [style={H box}] (1) at (-1.5, -0.5) {$H$};
		\node [style=none] (2) at (0.25, 0.75) {};
		\node [style=none] (3) at (0.75, -0.75) {\raisebox{2mm}{...}};
		\node [style={H box}] (4) at (-0.5, -0.5) {$H$};
		\node [style=none] (5) at (0.75, 0.75) {\raisebox{-2mm}{...}};
		\node [style=gsn] (6) at (0.75, -0) {$\beta$\nodepart{lower}$\alpha$};
		\node [style=none] (7) at (1.25, -0.75) {};
		\node [style=none] (8) at (0.25, -0.75) {};
		\node [style=none] (9) at (-1, 0.75) {\raisebox{-2mm}{...}};
		\node [style=none] (10) at (-1.5, 1) {};
		\node [style={H box}] (11) at (-0.5, 0.5) {$H^\dagger$};
		\node [style=none] (12) at (-0.5, -1) {};
		\node [style=none] (13) at (-1, -0.75) {\raisebox{2mm}{...}};
		\node [style=none] (14) at (1.25, 0.75) {};
		\node [style=rsn] (15) at (-1, -0) {$\alpha$\nodepart{lower}$\beta$};
		\node [style=none] (16) at (-1.5, -1) {};
		\node [style=none] (17) at (-0.5, 1) {};
		\node [style={H box}] (18) at (-1.5, 0.5) {$H^\dagger$};
	\end{pgfonlayer}
	\begin{pgfonlayer}{edgelayer}
		\draw (15) to (1);
		\draw (15) to (4);
		\draw (4) to (12.center);
		\draw (1) to (16.center);
		\draw (15) to (18);
		\draw (15) to (11);
		\draw (11) to (17.center);
		\draw (18) to (10.center);
		\draw (6) to (2.center);
		\draw (6) to (14.center);
		\draw (6) to (8.center);
		\draw (6) to (7.center);
	\end{pgfonlayer}
\end{tikzpicture}}
\endpgfgraphicnamed&(H2^\prime)\\
\multicolumn{3}{|c}{%
\beginpgfgraphicnamed{RGrelations/p1s}
\begin{tikzpicture}
	\begin{pgfonlayer}{nodelayer}
		\node [style=none] (0) at (0.5, -0.5) {};
		\node [style={H box}] (1) at (-1.25, -0) {$D$};
		\node [style=none] (2) at (-0.5, 0.5) {};
		\node [style=rn] (3) at (0.25, 0.25) {};
		\node [style=gn] (4) at (-0.25, -0.25) {};
		\node [style=none] (5) at (0.75, -0) {$=$};
		\node [style=none] (6) at (-1.25, 0.5) {};
		\node [style=none] (7) at (-0.75, -0) {$:=$};
		\node [style=none] (8) at (-1.25, -0.5) {};
		\node [style=none] (9) at (1.25, 0.75) {};
		\node [style=gn] (10) at (1.25, 0.25) {};
		\node [style=rn] (11) at (1.25, -0.25) {};
		\node [style=none] (12) at (1.25, -0.75) {};
		\node [style=none] (13) at (1.75, -0) {$=$};
		\node [style={H box}] (14) at (2.25, 0.25) {$H$};
		\node [style=none] (15) at (2.75, -0) {$=$};
		\node [style=none] (16) at (2.25, 0.75) {};
		\node [style={H box}] (17) at (3.25, 0.25) {$H^\dagger$};
		\node [style=none] (18) at (3.25, 0.75) {};
		\node [style={H box}] (19) at (3.25, -0.25) {$H^\dagger$};
		\node [style=none] (20) at (3.25, -0.75) {};
		\node [style=none] (21) at (2.25, -0.75) {};
		\node [style={H box}] (22) at (2.25, -0.25) {$H$};
	\end{pgfonlayer}
	\begin{pgfonlayer}{edgelayer}
		\draw [bend left=15, looseness=1.00] (4) to (2.center);
		\draw (6.center) to (1);
		\draw (10) to (9.center);
		\draw (1) to (8.center);
		\draw (4) to (3);
		\draw [bend left=15, looseness=1.00] (3) to (0.center);
		\draw (11) to (12.center);
		\draw [bend right=45, looseness=1.00] (10) to (11);
		\draw [bend left=45, looseness=1.00] (10) to (11);
		\draw (18.center) to (17);
		\draw (20.center) to (19);
		\draw (19) to (17);
		\draw (16.center) to (14);
		\draw (21.center) to (22);
		\draw (22) to (14);
	\end{pgfonlayer}
\end{tikzpicture}}
\endpgfgraphicnamed}&(P1)\\
\hline
\end{array}\]
\end{center}

  \caption{Qutrit ZX-calculus rewriting rules}\label{figure1}
\end{figure}

For the sake of simplicity, we will denote the
frequently used angles $\frac{2 \pi}{3}$ and   $\frac{4 \pi}{3}$ by $1$ and $2$ respectively.

The diagrams in Qutrit ZX-calculus have a standard interpretation (up to a non-zero scalar) $\llbracket \cdot \rrbracket$ in the Hilbert spaces:
\[
\left\llbracket%
\beginpgfgraphicnamed{Qutrits/generator_spider}
\begin{tikzpicture}
	\begin{pgfonlayer}{nodelayer}
		\node [style=none] (0) at (-1, 1) {};
		\node [style=none] (1) at (1, 1) {};
		\node [style=none] (2) at (-0.75, 0.75) {};
		\node [style=none] (3) at (-0.25, 0.75) {};
		\node [style=none] (4) at (0.75, 0.75) {};
		\node [style=none] (5) at (0.25, 0.5) {...};
		\node [style=gsn] (6) at (0, 0) {$\alpha$\nodepart{lower}$\beta$};
		\node [style=none] (7) at (0.25, -0.5) {...};
		\node [style=none] (8) at (-0.75, -0.75) {};
		\node [style=none] (9) at (-0.25, -0.75) {};
		\node [style=none] (10) at (0.75, -0.75) {};
		\node [style=none] (11) at (-1, -1) {};
		\node [style=none] (12) at (1, -1) {};
	\end{pgfonlayer}
	\begin{pgfonlayer}{edgelayer}
		\draw [style=braceedge] (12.center) to node[wire label, inner sep=5 pt]{$m$} (11.center);
		\draw [style=none, bend left=15, looseness=1.00] (8.center) to (6);
		\draw [style=none, bend left=15, looseness=1.00] (9.center) to (6);
		\draw [style=braceedge] (0.center) to node[wire label, inner sep=5 pt]{$n$} (1.center);
		\draw [style=none, bend left=15, looseness=1.00] (6) to (2.center);
		\draw [style=none, bend right=15, looseness=1.00] (10.center) to (6);
		\draw [style=none, bend right=15, looseness=1.00] (6) to (4.center);
		\draw [style=none, bend left=15, looseness=1.00] (6) to (3.center);
	\end{pgfonlayer}
\end{tikzpicture}}
\endpgfgraphicnamed\right\rrbracket=\ket{0}^{\otimes m}\bra{0}^{\otimes n}+e^{i\alpha}\ket{1}^{\otimes m}\bra{1}^{\otimes n}+e^{i\beta}\ket{2}^{\otimes m}\bra{2}^{\otimes n}
\]

\[
\left\llbracket%
\beginpgfgraphicnamed{Qutrits/generator_spider_gray}
\begin{tikzpicture}
	\begin{pgfonlayer}{nodelayer}
		\node [style=none] (0) at (1, -1) {};
		\node [style=none] (1) at (0.25, -0.5) {...};
		\node [style=none] (2) at (1, 1) {};
		\node [style=none] (3) at (0.75, -0.75) {};
		\node [style=none] (4) at (-0.25, -0.75) {};
		\node [style=none] (5) at (-0.25, 0.75) {};
		\node [style=none] (6) at (0.25, 0.5) {...};
		\node [style=none] (7) at (-0.75, -0.75) {};
		\node [style=rsn] (8) at (0, 0) {$\alpha$\nodepart{lower}$\beta$};
		\node [style=none] (9) at (-1, 1) {};
		\node [style=none] (10) at (-0.75, 0.75) {};
		\node [style=none] (11) at (-1, -1) {};
		\node [style=none] (12) at (0.75, 0.75) {};
	\end{pgfonlayer}
	\begin{pgfonlayer}{edgelayer}
		\draw [style=braceedge] (0.center) to node[wire label, inner sep=5 pt]{$m$} (11.center);
		\draw [style=none, bend left=15, looseness=1.00] (7.center) to (8);
		\draw [style=none, bend left=15, looseness=1.00] (4.center) to (8);
		\draw [style=braceedge] (9.center) to node[wire label, inner sep=5 pt]{$n$} (2.center);
		\draw [style=none, bend left=15, looseness=1.00] (8) to (10.center);
		\draw [style=none, bend right=15, looseness=1.00] (3.center) to (8);
		\draw [style=none, bend right=15, looseness=1.00] (8) to (12.center);
		\draw [style=none, bend left=15, looseness=1.00] (8) to (5.center);
	\end{pgfonlayer}
\end{tikzpicture}}
\endpgfgraphicnamed\right\rrbracket=\ket{+}^{\otimes m}\bra{+}^{\otimes n}+e^{i\alpha}\ket{\omega}^{\otimes m}\bra{\omega}^{\otimes n}+e^{i\beta}\ket{\bar{\omega}}^{\otimes m}\bra{\bar{\omega}}^{\otimes n}
\]

 \begin{equation*}
\left\llbracket%
\beginpgfgraphicnamed{RGgenerator/RGg_Hada}
\begin{tikzpicture}
	\begin{pgfonlayer}{nodelayer}
		\node [style={H box}] (0) at (0, -0) {$H$};
		\node [style=none] (1) at (0, 0.5) {};
		\node [style=none] (2) at (0, -0.5) {};
	\end{pgfonlayer}
	\begin{pgfonlayer}{edgelayer}
		\draw (1.center) to (0);
		\draw (2.center) to (0);
	\end{pgfonlayer}
\end{tikzpicture}}
\endpgfgraphicnamed\right\rrbracket=\ket{+}\bra{0}+ \ket{\omega}\bra{1}+\ket{\bar{\omega}}\bra{2}\quad\quad
\left\llbracket%
\beginpgfgraphicnamed{RGgenerator/RGg_Hadad}
\begin{tikzpicture}
	\begin{pgfonlayer}{nodelayer}
		\node [style={H box}] (0) at (0, -0) {$H^\dagger$};
		\node [style=none] (1) at (0, -0.5) {};
		\node [style=none] (2) at (0, 0.5) {};
	\end{pgfonlayer}
	\begin{pgfonlayer}{edgelayer}
		\draw (2.center) to (0);
		\draw (1.center) to (0);
	\end{pgfonlayer}
\end{tikzpicture}}
\endpgfgraphicnamed\right\rrbracket=\ket{0}\bra{+}+ \ket{1}\bra{\omega}+\ket{2}\bra{\bar{\omega}}
\end{equation*}

where $\bar{\omega}=e^{\frac{4}{3}\pi i}=\omega^2$, $ \ket{+}  =  \ket{0}+\ket{1}+\ket{2},  \ket{\omega}  =   \ket{0}+\omega \ket{1}+\bar{\omega}\ket{2}$, and $ \ket{\bar{\omega}}  =  \ket{0}+\bar{\omega}\ket{1}+\omega\ket{2}$.

For convenience, we also use the following matrix form:
\begin{trivlist}\item
    \begin{minipage}{0.495\textwidth}
      \begin{equation}
      \label{phasematrix}
  \left\llbracket%
\beginpgfgraphicnamed{RGgenerator/RGg_zph_ab}
\begin{tikzpicture}
	\begin{pgfonlayer}{nodelayer}
		\node [style=gsn] (0) at (0, -0) {$\alpha$\nodepart{lower}$\beta$};
		\node [style=none] (1) at (0, 0.5) {};
		\node [style=none] (2) at (0, -0.5) {};
	\end{pgfonlayer}
	\begin{pgfonlayer}{edgelayer}
		\draw (1.center) to (0);
		\draw (2.center) to (0);
	\end{pgfonlayer}
\end{tikzpicture}}
\endpgfgraphicnamed\right\rrbracket=
  \left(
\begin{array}{ccc}
 1 & 0 & 0 \\
 0 & e^{i\alpha} & 0\\
 0 & 0 & e^{i\beta}
\end{array}
\right)
\end{equation}
    \end{minipage}
   \begin{minipage}{0.495\textwidth}
      \begin{equation}
    \label{phasematrix3}
\left\llbracket%
\beginpgfgraphicnamed{RGgenerator/RGg_Hada}
\begin{tikzpicture}
	\begin{pgfonlayer}{nodelayer}
		\node [style={H box}] (0) at (0, -0) {$H$};
		\node [style=none] (1) at (0, 0.5) {};
		\node [style=none] (2) at (0, -0.5) {};
	\end{pgfonlayer}
	\begin{pgfonlayer}{edgelayer}
		\draw (1.center) to (0);
		\draw (2.center) to (0);
	\end{pgfonlayer}
\end{tikzpicture}}
\endpgfgraphicnamed\right\rrbracket=\begin{pmatrix}
        1 & 1 & 1 \\
        1 & \omega & \bar{\omega}\\
        1 & \bar{\omega} & \omega
 \end{pmatrix}
\end{equation}
    \end{minipage}
  \end{trivlist}
  \begin{equation}\label{phasematrix2}
\left\llbracket%
\beginpgfgraphicnamed{RGgenerator/RGg_xph_ab}
\begin{tikzpicture}
	\begin{pgfonlayer}{nodelayer}
		\node [style=rsn] (0) at (0, -0) {$\alpha$\nodepart{lower}$\beta$};
		\node [style=none] (1) at (0, 0.5) {};
		\node [style=none] (2) at (0, -0.5) {};
	\end{pgfonlayer}
	\begin{pgfonlayer}{edgelayer}
		\draw (1.center) to (0);
		\draw (2.center) to (0);
	\end{pgfonlayer}
\end{tikzpicture}}
\endpgfgraphicnamed\right\rrbracket=
  \left(
\begin{array}{ccc}
1+e^{i\alpha}+e^{i\beta} & 1+\bar{\omega} e^{i\alpha}+\omega e^{i\beta}& 1+\omega e^{i\alpha}+\bar{\omega}e^{i\beta} \\
1+\omega e^{i\alpha}+\bar{\omega}e^{i\beta}& 1+e^{i\alpha}+e^{i\beta} & 1+\bar{\omega}e^{i\alpha}+\omega e^{i\beta}\\
 1+\bar{\omega} e^{i\alpha}+\omega e^{i\beta}&1+\omega e^{i\alpha}+\bar{\omega}e^{i\beta} &1+e^{i\alpha}+e^{i\beta}
\end{array}
\right)
  \end{equation}

Like the qubit case, there are three important properties about the qutrit ZX-calculus, i.e. universality,  soundness, and completeness: Universality is about if there exists a ZX-calculus diagram for every corresponding linear map in Hilbert spaces under the standard interpretation. Soundness means that all the rules in the qutrit ZX-calculus have a correct standard interpretation in the  Hilbert spaces. Completeness is concerned with whether  an equation of diagrams can be derived in the ZX-calculus when their corresponding equation of linear maps under the standard interpretation holds true. 

Among these properties, universality is proved in \cite{BianWang2}. Soundness can easily be checked with the standard interpretation $\llbracket \cdot \rrbracket$.  Completeness has a negative answer in the qubit case \cite{Vladimir} for the overall QM, while for general pure qutrit QM, we do not have an answer at the moment, although we conjecture that  it is incomplete. 


\subsection{Qutrit stabilizer quantum mechanics in  the ZX-calculus}
In this subsection, we represent in ZX-calculus the qutrit stabilizer QM, which consists of  state preparations and measurements based on the computational basis $\{ \ket{0},  \ket{1},  \ket{2}\}$, as well as unitary operators belonging to the generalized Clifford group $\mathcal{C}_n$. Furthermore, we give the unique form for single qutrit Clifford operators.

 Firstly, the states and effects in computational basis can be represented as:

\[
\ket{0} =\left\llbracket%
\beginpgfgraphicnamed{RGgenerator/RGg_exd}
\begin{tikzpicture}
	\begin{pgfonlayer}{nodelayer}
		\node [style=none] (0) at (0, -0.25) {};
		\node [style=rn] (1) at (0, 0.25) {};
	\end{pgfonlayer}
	\begin{pgfonlayer}{edgelayer}
		\draw (0.center) to (1);
	\end{pgfonlayer}
\end{tikzpicture}}
\endpgfgraphicnamed\right\rrbracket, \quad
\ket{1} =\left\llbracket%
\beginpgfgraphicnamed{Qutrits/rdot1}
\begin{tikzpicture}
	\begin{pgfonlayer}{nodelayer}
		\node [style=rsn] (0) at (0, 0.25) {$2$\nodepart{lower}$1$};
		\node [style=none] (1) at (0, -0.25) {};
	\end{pgfonlayer}
	\begin{pgfonlayer}{edgelayer}
		\draw (1.center) to (0);
	\end{pgfonlayer}
\end{tikzpicture}}
\endpgfgraphicnamed\right\rrbracket, \quad
\ket{2} =\left\llbracket%
\beginpgfgraphicnamed{Qutrits/rdot2}
\begin{tikzpicture}
	\begin{pgfonlayer}{nodelayer}
		\node [style=rsn] (0) at (0, 0.25) {$1$\nodepart{lower}$2$};
		\node [style=none] (1) at (0, -0.25) {};
	\end{pgfonlayer}
	\begin{pgfonlayer}{edgelayer}
		\draw (1.center) to (0);
	\end{pgfonlayer}
\end{tikzpicture}}
\endpgfgraphicnamed\right\rrbracket, \quad
\bra{0}=\left\llbracket%
\beginpgfgraphicnamed{RGgenerator/RGg_ex}
\begin{tikzpicture}
	\begin{pgfonlayer}{nodelayer}
		\node [style=none] (0) at (0, 0.25) {};
		\node [style=rn] (1) at (0, -0.25) {};
	\end{pgfonlayer}
	\begin{pgfonlayer}{edgelayer}
		\draw (0.center) to (1);
	\end{pgfonlayer}
\end{tikzpicture}}
\endpgfgraphicnamed\right\rrbracket,\quad
\bra{1} =\left\llbracket%
\beginpgfgraphicnamed{Qutrits/rcdot1}
\begin{tikzpicture}
	\begin{pgfonlayer}{nodelayer}
		\node [style=none] (0) at (0, 0.25) {};
		\node [style=rsn] (1) at (0, -0.25) {$1$\nodepart{lower}$2$};
	\end{pgfonlayer}
	\begin{pgfonlayer}{edgelayer}
		\draw (0.center) to (1);
	\end{pgfonlayer}
\end{tikzpicture}}
\endpgfgraphicnamed\right\rrbracket, \quad
\bra{2} =\left\llbracket%
\beginpgfgraphicnamed{Qutrits/rcdot2}
\begin{tikzpicture}
	\begin{pgfonlayer}{nodelayer}
		\node [style=none] (0) at (0, 0.25) {};
		\node [style=rsn] (1) at (0, -0.25) {$2$\nodepart{lower}$1$};
	\end{pgfonlayer}
	\begin{pgfonlayer}{edgelayer}
		\draw (0.center) to (1);
	\end{pgfonlayer}
\end{tikzpicture}}
\endpgfgraphicnamed\right\rrbracket.
\]
The generators of the single  Clifford group $\mathcal{C}_n$  are denoted by the following diagrams:

\begin{equation}\label{clifgenerator}
\mathcal{S}=\left\llbracket%
\beginpgfgraphicnamed{Qutrits/phasegts}
\begin{tikzpicture}
	\begin{pgfonlayer}{nodelayer}
		\node [style=none] (0) at (0, 0.5) {};
		\node [style=none] (1) at (0, -0.5) {};
		\node [style=gsn] (2) at (0, 0) {$0$\nodepart{lower}$1$};
	\end{pgfonlayer}
	\begin{pgfonlayer}{edgelayer}
		\draw (0.center) to (2);
		\draw (1.center) to (2);
	\end{pgfonlayer}
\end{tikzpicture}}
\endpgfgraphicnamed\right\rrbracket,
\qquad
H=\left\llbracket%
\beginpgfgraphicnamed{RGgenerator/RGg_Hada}
\begin{tikzpicture}
	\begin{pgfonlayer}{nodelayer}
		\node [style={H box}] (0) at (0, -0) {$H$};
		\node [style=none] (1) at (0, 0.5) {};
		\node [style=none] (2) at (0, -0.5) {};
	\end{pgfonlayer}
	\begin{pgfonlayer}{edgelayer}
		\draw (1.center) to (0);
		\draw (2.center) to (0);
	\end{pgfonlayer}
\end{tikzpicture}}
\endpgfgraphicnamed\right\rrbracket, 
\qquad
\Lambda=\left\llbracket%
\beginpgfgraphicnamed{Qutrits/sumgt}
\begin{tikzpicture}
	\begin{pgfonlayer}{nodelayer}
		\node [style=none] (0) at (0.25, 0.5) {};
		\node [style=none] (1) at (-0.5, 0.5) {};
		\node [style=none] (2) at (-0.5, -0.5) {};
		\node [style=gn] (3) at (-0.5, 0.25) {};
		\node [style=rn] (4) at (0.25, -0.25) {};
		\node [style=none] (5) at (0.25, -0.5) {};
	\end{pgfonlayer}
	\begin{pgfonlayer}{edgelayer}
		\draw (2.center) to (3);
		\draw (4) to (0.center);
		\draw (5.center) to (4);
		\draw (3) to (1.center);
		\draw (3) to (4);
	\end{pgfonlayer}
\end{tikzpicture}}
\endpgfgraphicnamed\right\rrbracket.
\end{equation}

Now we characterise the diagrams for qutrit stabilizer QM.
  
\begin{theorem}
A diagram in the qutrit ZX-calculus stands for a linear map in the qutrit stabilizer QM under standard interpretation if and only if  all phase angles in this diagram are integer multiples of $\frac{2}{3}\pi$.
\end{theorem}

In the sequel, a diagram in which all phase angles are integer multiples of $\frac{2}{3}\pi$ will be called a  \textit{ stabilizer diagram}. Now we denote $\frac{2 \pi}{3}$ and  $\frac{4 \pi}{3}$  by $1$ and $2$ (or $-1$)  respectively, and let 
 $$\mathcal{P}= \{\nststile{1}{1}, \nststile{2}{2}\}, \quad \mathcal{N}=\{\nststile{0}{1},  \nststile{1}{0}, \nststile{0}{2},  \nststile{2}{0}\}, \quad \mathcal{M}=\{\nststile{0}{0},  \nststile{1}{2}, \nststile{2}{1}\}, \quad \mathcal{Q}=\mathcal{P} \cup \mathcal{N}, \quad \mathcal{A}=\mathcal{Q} \cup \mathcal{M},$$
where the symbol $\nststile{b}{a}$ is just the denotation of the pair $(a, b)$, instead of the fraction $\frac{a}{b}$. 
Then each green or red node in a stabilizer diagram has its phase angles denoted as elements in the set  $\mathcal{A}$.  Define the addition $+$ in  $\mathcal{A}$ as $\nststile{b}{a}+\nststile{d}{c} ~:=~ \nststile{b+d (mod 3)}{a+c (mod3)}$. Then $\mathcal{A}$ is an abelian group, $\mathcal{P}\cup \{\nststile{0}{0}\}~ \mbox{and} ~\mathcal{M}$ are subgroups of $\mathcal{A}$.


Unlike the quibit case, the fact that the ZX-calculus is complete for the single qutrit Clifford group $\mathcal{C}_1$ is far from trivial to be proved:

\begin{theorem}\label{uniformsin}
Each diagram of single qutrit Clifford operator can be uniquely rewritten into one of the following forms: 
\begin{trivlist}\item
    \begin{minipage}{0.295\textwidth}
      \begin{equation}
        \label{uniformsin1}
\beginpgfgraphicnamed{Qutrits/cliffnorformj}
\begin{tikzpicture}
	\begin{pgfonlayer}{nodelayer}
		\node [style=none] (0) at (0, -0.75) {};
		\node [style=gsn] (1) at (0, 0.5) {$a_1$\nodepart{lower}$a_2$};
		\node [style=none] (2) at (0, 1) {};
		\node [style=rsn] (3) at (0, -0.25) {$a_3$\nodepart{lower}$a_4$};
	\end{pgfonlayer}
	\begin{pgfonlayer}{edgelayer}
		\draw (2.center) to (1);
		\draw (0.center) to (1);
	\end{pgfonlayer}
\end{tikzpicture}}
\endpgfgraphicnamed    
\end{equation}
    \end{minipage}
   \begin{minipage}{0.295\textwidth}
      \begin{equation}
        \label{uniformsin2}
\beginpgfgraphicnamed{Qutrits/cliffnorformk}
\begin{tikzpicture}
	\begin{pgfonlayer}{nodelayer}
		\node [style=gsn] (0) at (0, 0.5) {$q_1$\nodepart{lower}$q_2$};
		\node [style=none] (1) at (0, -0.75) {};
		\node [style=rsn] (2) at (0, -0.25) {$a_5$\nodepart{lower}$a_6$};
		\node [style=none] (3) at (0, 1.75) {};
		\node [style=rsn] (4) at (0, 1.25) {$p_1$\nodepart{lower}$p_2$};
	\end{pgfonlayer}
	\begin{pgfonlayer}{edgelayer}
		\draw (3.center) to (0);
		\draw (1.center) to (0);
	\end{pgfonlayer}
\end{tikzpicture}}
\endpgfgraphicnamed    
\end{equation}
    \end{minipage}
  \begin{minipage}{0.295\textwidth}
      \begin{equation}
        \label{uniformsin3}
\beginpgfgraphicnamed{Qutrits/cliffnorformt}
\begin{tikzpicture}
	\begin{pgfonlayer}{nodelayer}
		\node [style=none] (0) at (0, -1) {};
		\node [style=none] (1) at (0, 1.75) {};
		\node [style=gsn] (2) at (0, 1.25) {$m_1$\nodepart{lower}$m_2$};
		\node [style=rsn] (3) at (0, 0.5) {$a_7$\nodepart{lower}$a_8$};
		\node [style=H box] (4) at (0, -0.5) {$H$};
		\node [style=H box] (5) at (0, 0) {$H$};
	\end{pgfonlayer}
	\begin{pgfonlayer}{edgelayer}
		\draw (1.center) to (3);
		\draw (0.center) to (3);
	\end{pgfonlayer}
\end{tikzpicture}}
\endpgfgraphicnamed    
\end{equation}
    \end{minipage}
  \end{trivlist}

where $\nststile{a_2}{a_1}, \nststile{a_4}{a_3}, \nststile{a_6}{a_5}, \nststile{a_8}{a_7} \in \mathcal{A}, \nststile{p_2}{p_1} \in \mathcal{P}, \nststile{q_2}{q_1} \in \mathcal{Q}, \nststile{m_2}{m_1} \in \mathcal{M}$. 
\end{theorem}


\section{Qutrit graph states in the ZX-calculus}

\subsection{Stabilizer state diagram and transformations of GS-LC diagrams }
To represent qutrit graph states in the ZX-calculus, we first show that the horizontal  nodes $H$ and $H^{\dagger}$  make sense when connected to two green nodes. 

\begin{lemma}\cite{GongWang}\label{controlz}
\begin{align*}
\beginpgfgraphicnamed{Qutrits/controlz}
\begin{tikzpicture}
	\begin{pgfonlayer}{nodelayer}
		\node [style=none] (0) at (-0.5, 1) {};
		\node [style=gn] (1) at (-0.5, 0.5) {};
		\node [style=H box] (2) at (0, 0) {$H$};
		\node [style=gn] (3) at (0.5, -0.5) {};
		\node [style=none] (4) at (-0.5, -1) {};
		\node [style=none] (5) at (0.5, 1) {};
		\node [style=none] (6) at (0.5, -1) {};
	\end{pgfonlayer}
	\begin{pgfonlayer}{edgelayer}
		\draw (5.center) to (6.center);
		\draw (0.center) to (4.center);
		\draw (1) to (2);
		\draw (2) to (3);
	\end{pgfonlayer}
\end{tikzpicture}}
\endpgfgraphicnamed=%
\beginpgfgraphicnamed{Qutrits/controlz2}
\begin{tikzpicture}
	\begin{pgfonlayer}{nodelayer}
		\node [style=none] (0) at (1.5, 0.75) {};
		\node [style={H box}] (1) at (1, -0.25) {$H$};
		\node [style=none] (2) at (1.5, -1.25) {};
		\node [style=none] (3) at (0.5, 0.75) {};
		\node [style=gn] (4) at (1.5, 0.25) {};
		\node [style=none] (5) at (0.5, -1.25) {};
		\node [style=gn] (6) at (0.5, -0.75) {};
	\end{pgfonlayer}
	\begin{pgfonlayer}{edgelayer}
		\draw (0.center) to (2.center);
		\draw (3.center) to (5.center);
		\draw (6) to (1);
		\draw (1) to (4);
	\end{pgfonlayer}
\end{tikzpicture}}
\endpgfgraphicnamed=: %
\beginpgfgraphicnamed{Qutrits/controlz3}
\begin{tikzpicture}
	\begin{pgfonlayer}{nodelayer}
		\node [style=none] (0) at (-0.5, 0.75) {};
		\node [style=gn] (1) at (-0.5, -0) {};
		\node [style={H box}] (2) at (0, -0) {$H$};
		\node [style=none] (3) at (0.5, 0.75) {};
		\node [style=gn] (4) at (0.5, -0) {};
		\node [style=none] (5) at (-0.5, -0.75) {};
		\node [style=none] (6) at (0.5, -0.75) {};
	\end{pgfonlayer}
	\begin{pgfonlayer}{edgelayer}
		\draw (3.center) to (6.center);
		\draw (0.center) to (5.center);
		\draw (1) to (2);
		\draw (2) to (4);
	\end{pgfonlayer}
\end{tikzpicture}}
\endpgfgraphicnamed,\quad\quad\quad\quad
%
\beginpgfgraphicnamed{Qutrits/controlzsquare}
\begin{tikzpicture}
	\begin{pgfonlayer}{nodelayer}
		\node [style=none] (0) at (1.5, 0.75) {};
		\node [style=none] (1) at (0.5, -1.25) {};
		\node [style=none] (2) at (1.5, -1.25) {};
		\node [style=gn] (3) at (0.5, 0.25) {};
		\node [style=gn] (4) at (1.5, -0.75) {};
		\node [style=none] (5) at (0.5, 0.75) {};
		\node [style={H box}] (6) at (1, -0.25) {$H^{\dagger}$};
	\end{pgfonlayer}
	\begin{pgfonlayer}{edgelayer}
		\draw (0.center) to (2.center);
		\draw (5.center) to (1.center);
		\draw (3) to (6);
		\draw (6) to (4);
	\end{pgfonlayer}
\end{tikzpicture}}
\endpgfgraphicnamed=%
\beginpgfgraphicnamed{Qutrits/controlzsquare2}
\begin{tikzpicture}
	\begin{pgfonlayer}{nodelayer}
		\node [style=gn] (0) at (0.5, -0.75) {};
		\node [style=gn] (1) at (1.5, 0.25) {};
		\node [style=none] (2) at (0.5, 0.75) {};
		\node [style=none] (3) at (0.5, -1.25) {};
		\node [style={H box}] (4) at (1, -0.25) {$H^{\dagger}$};
		\node [style=none] (5) at (1.5, 0.75) {};
		\node [style=none] (6) at (1.5, -1.25) {};
	\end{pgfonlayer}
	\begin{pgfonlayer}{edgelayer}
		\draw (5.center) to (6.center);
		\draw (2.center) to (3.center);
		\draw (0) to (4);
		\draw (4) to (1);
	\end{pgfonlayer}
\end{tikzpicture}}
\endpgfgraphicnamed=: %
\beginpgfgraphicnamed{Qutrits/controlzsquare3}
\begin{tikzpicture}
	\begin{pgfonlayer}{nodelayer}
		\node [style=none] (0) at (-0.5, 0.75) {};
		\node [style=gn] (1) at (0.5, -0) {};
		\node [style=gn] (2) at (-0.5, -0) {};
		\node [style=none] (3) at (-0.5, -0.75) {};
		\node [style=none] (4) at (0.5, 0.75) {};
		\node [style=none] (5) at (0.5, -0.75) {};
		\node [style={H box}] (6) at (0, -0) {$H^{\dagger}$};
	\end{pgfonlayer}
	\begin{pgfonlayer}{edgelayer}
		\draw (4.center) to (5.center);
		\draw (0.center) to (3.center);
		\draw (2) to (6);
		\draw (6) to (1);
	\end{pgfonlayer}
\end{tikzpicture}}
\endpgfgraphicnamed.
\end{align*}
\end{lemma}

\begin{proof}
We will only prove the first equation.
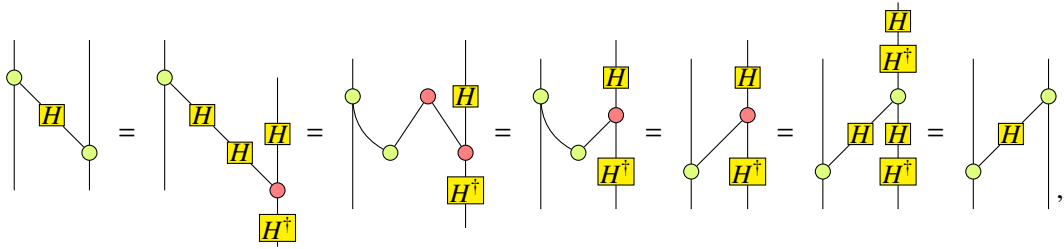
\begin{align*}
\beginpgfgraphicnamed{Qutrits/controlzprf}
\begin{tikzpicture}
	\begin{pgfonlayer}{nodelayer}
		\node [style=none] (0) at (-6, -0) {};
		\node [style=gn] (1) at (-6, 0.5) {};
		\node [style=gn] (2) at (-7, 1.5) {};
		\node [style=none] (3) at (-7, 2) {};
		\node [style=none] (4) at (-7, -0) {};
		\node [style=none] (5) at (-6, 2) {};
		\node [style={H box}] (6) at (-6.5, 1) {$H$};
		\node [style={H box}] (7) at (-4.5, 1) {$H$};
		\node [style=none] (8) at (-3.5, -0.75) {};
		\node [style=none] (9) at (-5, -0) {};
		\node [style=gn] (10) at (-5, 1.5) {};
		\node [style=none] (11) at (-5, 2) {};
		\node [style=none] (12) at (-3.5, 1.5) {};
		\node [style=none] (13) at (-3, 0.75) {$=$};
		\node [style=none] (14) at (1.5, 0.75) {$=$};
		\node [style=none] (15) at (3.25, 0.75) {$=$};
		\node [style=none] (16) at (-5.5, 0.75) {$=$};
		\node [style=none] (17) at (-0.5, 0.75) {$=$};
		\node [style={H box}] (18) at (-4, 0.5) {$H$};
		\node [style={H box}] (19) at (-3.5, 0.75) {$H$};
		\node [style={H box}] (20) at (-3.5, -0.5) {$H^\dagger$};
		\node [style={H box}] (21) at (-1, -0) {$H^\dagger$};
		\node [style={H box}] (22) at (-1, 1.25) {$H$};
		\node [style=none] (23) at (-1, -0.5) {};
		\node [style=none] (24) at (-1, 2) {};
		\node [style=none] (25) at (-2.5, 1.75) {};
		\node [style=gn] (26) at (-2.5, 1.25) {};
		\node [style=none] (27) at (-2.5, -0.25) {};
		\node [style=rn] (28) at (-1.5, 1.25) {};
		\node [style=gn] (29) at (-2, 0.5) {};
		\node [style=none] (30) at (-2.5, 1.25) {};
		\node [style=rn] (31) at (-1, 0.5) {};
		\node [style=rn] (32) at (-3.5, -0) {};
		\node [style=rn] (33) at (1, 1) {};
		\node [style={H box}] (34) at (1, 0.25) {$H^\dagger$};
		\node [style=gn] (35) at (0.5, 0.5) {};
		\node [style={H box}] (36) at (1, 1.5) {$H$};
		\node [style=none] (37) at (0, 1.75) {};
		\node [style=gn] (38) at (0, 1.25) {};
		\node [style=none] (39) at (0, 1.25) {};
		\node [style=none] (40) at (1, -0.25) {};
		\node [style=none] (41) at (0, -0.25) {};
		\node [style=none] (42) at (1, 2) {};
		\node [style=gn] (43) at (2, 0.25) {};
		\node [style=none] (44) at (2.75, -0.25) {};
		\node [style=none] (45) at (2, 1.75) {};
		\node [style=rn] (46) at (2.75, 1) {};
		\node [style={H box}] (47) at (2.75, 1.5) {$H$};
		\node [style=none] (48) at (2.75, 2) {};
		\node [style=none] (49) at (2, -0.25) {};
		\node [style={H box}] (50) at (2.75, 0.25) {$H^\dagger$};
		\node [style=none] (51) at (3.75, 1.75) {};
		\node [style=none] (52) at (4.75, 2.5) {};
		\node [style=gn] (53) at (4.75, 1.25) {};
		\node [style=none] (54) at (3.75, -0.25) {};
		\node [style=none] (55) at (4.75, -0.25) {};
		\node [style={H box}] (56) at (4.25, 0.75) {$H$};
		\node [style=gn] (57) at (3.75, 0.25) {};
		\node [style={H box}] (58) at (4.75, 0.75) {$H$};
		\node [style={H box}] (59) at (4.75, 0.25) {$H^\dagger$};
		\node [style={H box}] (60) at (4.75, 1.75) {$H^\dagger$};
		\node [style={H box}] (61) at (4.75, 2.25) {$H$};
		\node [style=none] (62) at (5.25, 0.75) {$=$};
		\node [style=none] (63) at (5.75, 1.75) {};
		\node [style=none] (64) at (6.75, 1.75) {};
		\node [style=gn] (65) at (6.75, 1.25) {};
		\node [style=none] (66) at (5.75, -0.25) {};
		\node [style=none] (67) at (6.75, -0.25) {};
		\node [style={H box}] (68) at (6.25, 0.75) {$H$};
		\node [style=gn] (69) at (5.75, 0.25) {};
	\end{pgfonlayer}
	\begin{pgfonlayer}{edgelayer}
		\draw (5.center) to (0.center);
		\draw (3.center) to (4.center);
		\draw (2) to (6);
		\draw (6) to (1);
		\draw (12.center) to (8.center);
		\draw (11.center) to (9.center);
		\draw (10) to (7);
		\draw (24.center) to (23.center);
		\draw (25.center) to (27.center);
		\draw [bend left, looseness=1.00] (29) to (30.center);
		\draw (29) to (28);
		\draw (28) to (31);
		\draw (7) to (32);
		\draw (42.center) to (40.center);
		\draw (37.center) to (41.center);
		\draw [bend left, looseness=1.00] (35) to (39.center);
		\draw (33) to (35);
		\draw (48.center) to (44.center);
		\draw (45.center) to (49.center);
		\draw (43) to (46);
		\draw (52.center) to (55.center);
		\draw (51.center) to (54.center);
		\draw (57) to (56);
		\draw (56) to (53);
		\draw (64.center) to (67.center);
		\draw (63.center) to (66.center);
		\draw (69) to (68);
		\draw (68) to (65);
	\end{pgfonlayer}
\end{tikzpicture}}
\endpgfgraphicnamed,
 \end{align*}
where we used (H2) for the first and the fifth equalities, and (P1) for the second equality. 
\end{proof}

Now qutrit graph states have a nice representation in the ZX-calculus.

\begin{lemma}\cite{GongWang}
A qutrit graph state $\ket{G}$, where $G = (E;V)$ is an n-vertex graph, is represented in the graphical
calculus as follows:
\begin{itemize}
 \item  for each vertex $v \in V$, a green node with one output, and
\item  for each $1$-weighted edge $\{u,v\}\in E$, an $H$ node connected to the green nodes representing vertices $u$
and $v$,
\item  for each $2$-weighted edge $\{u,v\}\in E$, an $H^{\dagger}$ node connected to the green nodes representing vertices $u$
and $v$.

\end{itemize}
\end{lemma}
A graph state $\ket{G}$ is also denoted by the diagram %
\beginpgfgraphicnamed{Qutrits/grapstg}
\begin{tikzpicture}
	\begin{pgfonlayer}{nodelayer}
		\node [style=none] (0) at (-1, 0) {};
		\node [style=none] (1) at (-0.75, 0) {};
		\node [style=none] (2) at (0, -0.25) {$\cdots$};
		\node [style=none] (3) at (0, 0.25) {$G$};
		\node [style=none] (4) at (0.75, -0.25) {};
		\node [style=none] (5) at (-0.75, -0.25) {};
		\node [style=none] (6) at (1, 0) {};
		\node [style=none] (7) at (0.75, 0) {};
	\end{pgfonlayer}
	\begin{pgfonlayer}{edgelayer}
		\draw [bend left=75, looseness=0.75] (0.center) to (6.center);
		\draw (0.center) to (6.center);
		\draw (1.center) to (5.center);
		\draw (7.center) to (4.center);
	\end{pgfonlayer}
\end{tikzpicture}}
\endpgfgraphicnamed.

\begin{definition}\cite{Miriam2}
A diagram in the stabilizer ZX-calculus is called a GS-LC diagram if it consists of a graph state diagram with arbitrary single-qutrit Clifford operators applied to each output. These associated Clifford operators are called vertex operators.
\end{definition}
An $n$-qutrit GS-LC diagram is represented as
 \begin{equation*}
\beginpgfgraphicnamed{Qutrits/gslc}
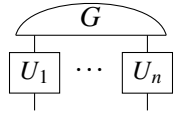
\begin{tikzpicture}
	\begin{pgfonlayer}{nodelayer}
		\node [style=none] (0) at (-1, 0) {};
		\node [style=none] (1) at (-0.75, -1) {};
		\node [style=none] (2) at (1, 0) {};
		\node [style=none] (3) at (0, 0.25) {$G$};
		\node [style=none] (4) at (0.75, 0) {};
		\node [style=none] (5) at (-0.75, 0) {};
		\node [style=square box] (6) at (-0.75, -0.5) {$U_1$};
		\node [style=none] (7) at (0, -0.5) {$\cdots$};
		\node [style=none] (8) at (0.75, -1) {};
		\node [style=square box] (9) at (0.75, -0.5) {$U_n$};
	\end{pgfonlayer}
	\begin{pgfonlayer}{edgelayer}
		\draw [bend left=75, looseness=0.75] (0.center) to (2.center);
		\draw (0.center) to (2.center);
		\draw (5.center) to (1.center);
		\draw (4.center) to (8.center);
	\end{pgfonlayer}
\end{tikzpicture}}
\endpgfgraphicnamed
     \end{equation*}
where $G = (V,E)$ is a graph  and $U_v \in\mathcal{C}_1$ for $v\in V$.

\begin{theorem}\label{gandgst1}\cite{GongWang}
Let $G = (V,E)$ be a graph with adjacency matrix $\Gamma$ and  $G\ast_1  v $ be the graph that
results from applying to $G$ a $1$-local complementation about some $v\in V$. Then the corresponding graph states 
$\ket{G} $ and $\ket{G\ast_1  v} $ are related as follows:
\begin{equation}\label{locom1}
\beginpgfgraphicnamed{Qutrits/gandg1}
\begin{tikzpicture}
	\begin{pgfonlayer}{nodelayer}
		\node [style=none] (0) at (-2.25, 0.25) {$G$};
		\node [style=none] (1) at (-3, 0) {};
		\node [style=none] (2) at (-3, -0.5) {};
		\node [style=none] (3) at (-3.25, 0) {};
		\node [style=none] (4) at (-1.25, 0) {};
		\node [style=none] (5) at (-0.5, 0) {$=$};
		\node [style=none] (6) at (-1.5, 0) {};
		\node [style=none] (7) at (-1.5, -0.5) {};
		\node [style=none] (8) at (-2.25, -0.25) {$\cdots$};
		\node [style=none] (9) at (0.5, 0) {};
		\node [style=none] (10) at (0.25, 0) {};
		\node [style=none] (11) at (0.5, -1) {};
		\node [style=none] (12) at (1.75, -1) {};
		\node [style=none] (13) at (3.25, 0) {};
		\node [style=none] (14) at (1.75, 0) {};
		\node [style=none] (15) at (1.25, -0.5) {$\cdots$};
		\node [style=none] (16) at (1.75, 0.25) {$G\ast_1  v$};
		\node [style=gsn] (17) at (0.5, -0.5) {$a_1$\nodepart{lower}$b_1$};
		\node [style=rsn] (18) at (1.75, -0.5) {$1$\nodepart{lower}$1$};
		\node [style=none] (19) at (2.25, -0.5) {$\cdots$};
		\node [style=none] (20) at (3, -1) {};
		\node [style=none] (21) at (3, 0) {};
		\node [style=gsn] (22) at (3, -0.5) {$a_n$\nodepart{lower}$b_n$};
	\end{pgfonlayer}
	\begin{pgfonlayer}{edgelayer}
		\draw [bend left=75, looseness=0.75] (3.center) to (4.center);
		\draw (3.center) to (4.center);
		\draw (1.center) to (2.center);
		\draw (6.center) to (7.center);
		\draw [bend left=75, looseness=0.75] (10.center) to (13.center);
		\draw (10.center) to (13.center);
		\draw (9.center) to (11.center);
		\draw (14.center) to (12.center);
		\draw (21.center) to (20.center);
	\end{pgfonlayer}
\end{tikzpicture}}
\endpgfgraphicnamed
\end{equation}
where for $1\leqslant i \leqslant n, ~i\neq v$,
 \begin{equation*}
\nststile{b_i}{a_i} =  \left\{\begin{array}{l}
    \nststile{2}{2}, ~\mbox{if}~ \Gamma_{iv}\neq 0 \vspace{0.2cm}\\
    \nststile{0}{0}, ~\mbox{if} ~\Gamma_{iv}= 0\\
    \end{array}\right. 
     \end{equation*}
\end{theorem}

\begin{corollary}\label{gandgst2}
Let $G = (V,E)$ be a graph with adjacency matrix $\Gamma$ and  $G\ast_2  v $ be the graph that
results from applying to $G$ a $2$-local complementation about some $v\in V$. Then the corresponding graph states 
$\ket{G} $ and $\ket{G\ast_2  v} $ are related as follows:
\begin{equation}\label{locom2}
\beginpgfgraphicnamed{Qutrits/gandg2}
\begin{tikzpicture}
	\begin{pgfonlayer}{nodelayer}
		\node [style=none] (0) at (-3.25, 0) {};
		\node [style=none] (1) at (3, -1) {};
		\node [style=rsn] (2) at (1.75, -0.5) {$2$\nodepart{lower}$2$};
		\node [style=none] (3) at (1.75, 0) {};
		\node [style=none] (4) at (-3, -0.5) {};
		\node [style=none] (5) at (-1.25, 0) {};
		\node [style=none] (6) at (-2.25, 0.25) {$G$};
		\node [style=none] (7) at (-0.5, 0) {$=$};
		\node [style=none] (8) at (2.25, -0.5) {$\cdots$};
		\node [style=gsn] (9) at (0.5, -0.5) {$a_1$\nodepart{lower}$b_1$};
		\node [style=none] (10) at (0.5, 0) {};
		\node [style=none] (11) at (-1.5, 0) {};
		\node [style=none] (12) at (-3, 0) {};
		\node [style=none] (13) at (1.75, 0.25) {$G\ast_2  v$};
		\node [style=none] (14) at (0.5, -1) {};
		\node [style=none] (15) at (3, 0) {};
		\node [style=none] (16) at (1.25, -0.5) {$\cdots$};
		\node [style=none] (17) at (-2.25, -0.25) {$\cdots$};
		\node [style=none] (18) at (3.25, 0) {};
		\node [style=none] (19) at (0.25, 0) {};
		\node [style=none] (20) at (1.75, -1) {};
		\node [style=none] (21) at (-1.5, -0.5) {};
		\node [style=gsn] (22) at (3, -0.5) {$a_n$\nodepart{lower}$b_n$};
	\end{pgfonlayer}
	\begin{pgfonlayer}{edgelayer}
		\draw [bend left=75, looseness=0.75] (0.center) to (5.center);
		\draw (0.center) to (5.center);
		\draw (12.center) to (4.center);
		\draw (11.center) to (21.center);
		\draw [bend left=75, looseness=0.75] (19.center) to (18.center);
		\draw (19.center) to (18.center);
		\draw (10.center) to (14.center);
		\draw (3.center) to (20.center);
		\draw (15.center) to (1.center);
	\end{pgfonlayer}
\end{tikzpicture}}
\endpgfgraphicnamed
\end{equation}
where for $1\leqslant i \leqslant n, ~i\neq v$,
 \begin{equation*}
\nststile{b_i}{a_i} =  \left\{\begin{array}{l}
    \nststile{1}{1}, ~\mbox{if}~ \Gamma_{iv}\neq 0 \vspace{0.2cm}\\
    \nststile{0}{0}, ~\mbox{if} ~\Gamma_{iv}= 0\\
    \end{array}\right. 
     \end{equation*}
\end{corollary}
\begin{proof}
Note that $G\ast_2  v=(G\ast_1  v)\ast_1  v$.
\end{proof}

\begin{theorem}\label{gandgst3}
Let $G = (V,E)$ be a graph with adjacency matrix $\Gamma$ and  $G\circ_2  v $ be the graph that
results from applying  to $G$ the transformation $\circ_2  v $ about some $v\in V$. Then the corresponding graph states 
$\ket{G} $ and $\ket{G\circ_2  v} $ are related as follows:
\begin{equation}\label{locom3}
\beginpgfgraphicnamed{Qutrits/gandg3}
\begin{tikzpicture}
	\begin{pgfonlayer}{nodelayer}
		\node [style=none] (0) at (2.25, -0.5) {$\cdots$};
		\node [style=none] (1) at (-3.25, 0) {};
		\node [style=none] (2) at (3.25, 0) {};
		\node [style=none] (3) at (-1.5, 0) {};
		\node [style=none] (4) at (-1.5, -0.5) {};
		\node [style=none] (5) at (-1.25, 0) {};
		\node [style=none] (6) at (1.75, -1) {};
		\node [style=none] (7) at (-3, 0) {};
		\node [style=none] (8) at (3, 0) {};
		\node [style=none] (9) at (3, -1) {};
		\node [style=none] (10) at (1.75, 0.25) {$G\circ_2  v$};
		\node [style=none] (11) at (-0.5, 0) {$=$};
		\node [style=none] (12) at (-2.25, -0.25) {$\cdots$};
		\node [style=none] (13) at (1.25, -0.5) {$\cdots$};
		\node [style=none] (14) at (0.5, -1) {};
		\node [style=none] (15) at (-2.25, 0.25) {$G$};
		\node [style=none] (16) at (1.75, 0) {};
		\node [style=none] (17) at (0.25, 0) {};
		\node [style=none] (18) at (-3, -0.5) {};
		\node [style=none] (19) at (0.5, 0) {};
		\node [style=H box] (20) at (1.75, -0.25) {$H$};
		\node [style=H box] (21) at (1.75, -0.75) {$H$};
	\end{pgfonlayer}
	\begin{pgfonlayer}{edgelayer}
		\draw [bend left=75, looseness=0.75] (1.center) to (5.center);
		\draw (1.center) to (5.center);
		\draw (7.center) to (18.center);
		\draw (3.center) to (4.center);
		\draw [bend left=75, looseness=0.75] (17.center) to (2.center);
		\draw (17.center) to (2.center);
		\draw (19.center) to (14.center);
		\draw (16.center) to (6.center);
		\draw (8.center) to (9.center);
	\end{pgfonlayer}
\end{tikzpicture}}
\endpgfgraphicnamed
\end{equation}
where the Hadamard nodes are on the output of the vertex $v$.

\end{theorem}


With these equivalences of transformations of GS-LC diagrams, it can be shown that each qutrit stabilizer state diagram is equal to some GS-LC diagram.

\begin{theorem}\label{stabtogslc1}
In the qutrit ZX-calculus, each qutrit stabilizer state diagram can be rewritten into some GS-LC diagram.\end{theorem}

\subsection{Reduced GS-LC diagrams}
Further to GS-LC diagrams, we can define a more reduced form as in \cite{Miriam2}.

\begin{definition}\label{rgslcdef}

A stabilizer state diagram is called to be in a reduced GS-LC (or rGS-LC) form if it is a GS-LC diagram satisfying the following conditions:
\begin{itemize}
\item All vertex operators belong to the set
\begin{equation}\label{rgslc1}
  \mathcal{R}=\left\{\quad %
\beginpgfgraphicnamed{Qutrits/singlephases}
\begin{tikzpicture}
	\begin{pgfonlayer}{nodelayer}
		\node [style=gsn] (0) at (-1.5, 0.25) {$0$\nodepart{lower}$1$};
		\node [style=gsn] (1) at (-3, 0.25) {$2$\nodepart{lower}$1$};
		\node [style=none] (2) at (0, -0.5) {};
		\node [style=gsn] (3) at (-2, 0.25) {$2$\nodepart{lower}$2$};
		\node [style=none] (4) at (-4, -0.5) {};
		\node [style=gsn] (5) at (-3.5, 0.25) {$1$\nodepart{lower}$2$};
		\node [style=gsn] (6) at (-2.5, 0.25) {$1$\nodepart{lower}$1$};
		\node [style=gsn] (7) at (-1, 0.25) {$1$\nodepart{lower}$0$};
		\node [style=none] (8) at (-0.5, -0.5) {};
		\node [style=gsn] (9) at (-0.5, 0.25) {$0$\nodepart{lower}$2$};
		\node [style=none] (10) at (-1.5, -0.5) {};
		\node [style=none] (11) at (-3.5, -0.5) {};
		\node [style=none] (12) at (-2.5, -0.5) {};
		\node [style=none] (13) at (-3, -0.5) {};
		\node [style=gsn] (14) at (0, 0.25) {$2$\nodepart{lower}$0$};
		\node [style=none] (15) at (0.5, -0.5) {};
		\node [style=none] (16) at (-1, -0.5) {};
		\node [style=none] (17) at (-2, -0.5) {};
		\node [style=none] (18) at (-3.5, 1) {};
		\node [style=none] (19) at (-2.5, 1) {};
		\node [style=none] (20) at (-1.5, 1) {};
		\node [style=none] (21) at (0, 1) {};
		\node [style=none] (22) at (0.5, 1) {};
		\node [style=none] (23) at (-4, 1) {};
		\node [style=none] (24) at (-3, 1) {};
		\node [style=none] (25) at (-0.5, 1) {};
		\node [style=none] (26) at (-2, 1) {};
		\node [style=none] (27) at (-1, 1) {};
		\node [style=rsn] (28) at (0.5, 0) {$2$\nodepart{lower}$2$};
		\node [style=gsn] (29) at (0.5, 0.5) {$2$\nodepart{lower}$2$};
		\node [style=none] (30) at (1, -0.5) {};
		\node [style=rsn] (31) at (1, 0) {$2$\nodepart{lower}$2$};
		\node [style=none] (32) at (1, 1) {};
		\node [style=gsn] (33) at (1, 0.5) {$0$\nodepart{lower}$1$};
		\node [style=none] (34) at (1.5, -0.5) {};
		\node [style=rsn] (35) at (1.5, 0) {$2$\nodepart{lower}$2$};
		\node [style=none] (36) at (1.5, 1) {};
		\node [style=gsn] (37) at (1.5, 0.5) {$1$\nodepart{lower}$0$};
		\node [style=none] (38) at (1, 0.5) {};
		\node [style=rsn] (39) at (2, 0) {$1$\nodepart{lower}$1$};
		\node [style=gsn] (40) at (2, 0.5) {$1$\nodepart{lower}$1$};
		\node [style=none] (41) at (2, -0.5) {};
		\node [style=none] (42) at (2, 1) {};
		\node [style=rsn] (43) at (2.5, 0) {$1$\nodepart{lower}$1$};
		\node [style=none] (44) at (2.5, -0.5) {};
		\node [style=gsn] (45) at (2.5, 0.5) {$0$\nodepart{lower}$2$};
		\node [style=none] (46) at (2.5, 1) {};
		\node [style=rsn] (47) at (3, 0) {$1$\nodepart{lower}$1$};
		\node [style=none] (48) at (3, -0.5) {};
		\node [style=gsn] (49) at (3, 0.5) {$2$\nodepart{lower}$0$};
		\node [style=none] (50) at (3, 1) {};
	\end{pgfonlayer}
	\begin{pgfonlayer}{edgelayer}
		\draw (11.center) to (5);
		\draw (13.center) to (1);
		\draw (12.center) to (6);
		\draw (17.center) to (3);
		\draw (10.center) to (0);
		\draw (16.center) to (7);
		\draw (8.center) to (9);
		\draw (2.center) to (14);
		\draw (18.center) to (5);
		\draw (24.center) to (1);
		\draw (19.center) to (6);
		\draw (26.center) to (3);
		\draw (20.center) to (0);
		\draw (27.center) to (7);
		\draw (25.center) to (9);
		\draw (21.center) to (14);
		\draw (23.center) to (4.center);
		\draw (22.center) to (15.center);
		\draw (32.center) to (30.center);
		\draw (36.center) to (34.center);
		\draw (42.center) to (41.center);
		\draw (46.center) to (44.center);
		\draw (50.center) to (48.center);
	\end{pgfonlayer}
\end{tikzpicture}}
\endpgfgraphicnamed\quad\right\}.
  \end{equation}

\item Two adjacent vertices must not both have vertex operators including red nodes.
\end{itemize}
\end{definition}

\begin{theorem}\label{gslctorgslc}
In the qutrit ZX-calculus, each qutrit stabilizer state diagram can be rewritten into some rGS-LC diagram.
  \end{theorem}

\subsection{Transformations of rGS-LC diagrams}

 In this subsection we show how to  transform one rGS-LC diagrams into another rGS-LC diagrams. Note that we call the 
graphical transformation $\circ_2  v $ a doubling-neighbour-edge transformation.
\begin{lemma}\label{equivrgslc}
Suppose there is a rGS-LC diagram which has a pair of neighbouring qutrits $a$ and $b$  as follows:\begin{equation}\label{equivrgslcs1}
\beginpgfgraphicnamed{Qutrits/equivrgslct1}
\begin{tikzpicture}
	\begin{pgfonlayer}{nodelayer}
		\node [style=none] (0) at (-1, 1.25) {};
		\node [style=gsn] (1) at (-0.75, -0.25) {$a_1$\nodepart{lower}$a_2$};
		\node [style=none] (2) at (-1.75, 1) {$\cdots$};
		\node [style=none] (3) at (-2, 1.25) {};
		\node [style=gn] (4) at (-0.75, 0.5) {};
		\node [style=gsn] (5) at (-1.75, 0) {$q_1$\nodepart{lower}$q_2$};
		\node [style=none] (6) at (-0.75, -1.25) {};
		\node [style=none] (7) at (-1.75, -1.25) {};
		\node [style=none] (8) at (-0.75, 1) {$\cdots$};
		\node [style=none] (9) at (-0.5, 1.25) {};
		\node [style=none] (10) at (-1.5, 1.25) {};
		\node [style=none] (11) at (-0.5, 0.5) {$b$};
		\node [style=H box] (12) at (-1.25, 0.5) {$h$};
		\node [style=none] (13) at (-2, 0.5) {$a$};
		\node [style=gn] (14) at (-1.75, 0.5) {};
		\node [style=rsn] (15) at (-1.75, -0.75) {$p_1$\nodepart{lower}$p_1$};
		\node [style=none] (16) at (0, 0) {$=$};
		\node [style=H box] (17) at (1.25, 0.5) {$h$};
		\node [style=none] (18) at (1.75, 1) {$\cdots$};
		\node [style=gsn] (19) at (1.75, -0.25) {$a_1$\nodepart{lower}$a_2$};
		\node [style=none] (20) at (2, 0.5) {$b$};
		\node [style=none] (21) at (0.75, 1) {$\cdots$};
		\node [style=rsn] (22) at (0.75, -1.5) {$p_1$\nodepart{lower}$p_1$};
		\node [style=none] (23) at (0.5, 1.25) {};
		\node [style=gn] (24) at (1.75, 0.5) {};
		\node [style=none] (25) at (1.75, -1.25) {};
		\node [style=none] (26) at (2, 1.25) {};
		\node [style=none] (27) at (0.75, -2) {};
		\node [style=gn] (28) at (0.75, 0.5) {};
		\node [style=none] (29) at (1, 1.25) {};
		\node [style=none] (30) at (0.5, 0.5) {$a$};
		\node [style=none] (31) at (1.5, 1.25) {};
		\node [style=gsn] (32) at (0.75, 0) {$m_1$\nodepart{lower}$m_2$};
		\node [style=gsn] (33) at (0.75, -0.75) {$p_1$\nodepart{lower}$p_1$};
	\end{pgfonlayer}
	\begin{pgfonlayer}{edgelayer}
		\draw [bend left=15, looseness=1.00] (10.center) to (14);
		\draw [bend right=15, looseness=1.00] (3.center) to (14);
		\draw [bend right=15, looseness=1.00] (0.center) to (4);
		\draw [bend left=15, looseness=1.00] (9.center) to (4);
		\draw (14) to (4);
		\draw (14) to (7.center);
		\draw (4) to (6.center);
		\draw [bend left=15, looseness=1.00] (29.center) to (28);
		\draw [bend right=15, looseness=1.00] (23.center) to (28);
		\draw [bend right=15, looseness=1.00] (31.center) to (24);
		\draw [bend left=15, looseness=1.00] (26.center) to (24);
		\draw (28) to (24);
		\draw (28) to (27.center);
		\draw (24) to (25.center);
	\end{pgfonlayer}
\end{tikzpicture}}
\endpgfgraphicnamed
\end{equation}
where  $\nststile{a_2}{a_1}\in \mathcal{A}, \nststile{p_1}{p_1}\in \mathcal{P}, \nststile{m_2}{m_1}\in \mathcal{M}, \nststile{q_2}{q_1} = \nststile{p_1}{p_1} +\nststile{m_2}{m_1}$, $h$ stands for either an $H$ node or an $H^{\dagger}$ node.   Then a  rGS-LC diagram with the  following pair can be obtained by  performing firstly a $p_1$-local complementation about $b$, followed by a
$(-p_1)$-local complementation about $a$, and possibly  a further  
$(-p_1)$-local complementation about $b$, in addition with some doubling-neighbour-edge operations $\circ_2  a(b) $ and copying operations on red nodes with phase angles in $ \mathcal{M}$:
\begin{equation}\label{equivrgslcs2}
\beginpgfgraphicnamed{Qutrits/equivrgslct2}
\begin{tikzpicture}
	\begin{pgfonlayer}{nodelayer}
		\node [style=gn] (0) at (-0.5, 0.5) {};
		\node [style=none] (1) at (-0.75, 1.25) {};
		\node [style=none] (2) at (0.5, -2) {};
		\node [style=gn] (3) at (0.5, 0.5) {};
		\node [style=rsn] (4) at (0.5, -1.5) {$p^{\prime}_1$\nodepart{lower}$p^{\prime}_1$};
		\node [style=none] (5) at (-0.75, 0.5) {$a$};
		\node [style=none] (6) at (-0.5, -1.25) {};
		\node [style=none] (7) at (0.5, 1) {$\cdots$};
		\node [style=gsn] (8) at (-0.5, -0.25) {$a^{\prime}_1$\nodepart{lower}$a^{\prime}_2$};
		\node [style=gsn] (9) at (0.5, -0.75) {$p^{\prime}_1$\nodepart{lower}$p^{\prime}_1$};
		\node [style=none] (10) at (-0.5, 1) {$\cdots$};
		\node [style=none] (11) at (0.75, 1.25) {};
		\node [style=gsn] (12) at (0.5, 0) {$m^{\prime}_1$\nodepart{lower}$m^{\prime}_2$};
		\node [style=none] (13) at (-0.25, 1.25) {};
		\node [style=none] (14) at (0.25, 1.25) {};
		\node [style=H box] (15) at (0, 0.5) {$h$};
		\node [style=none] (16) at (0.75, 0.5) {$b$};
	\end{pgfonlayer}
	\begin{pgfonlayer}{edgelayer}
		\draw [bend right=15, looseness=1.00] (14.center) to (3);
		\draw [bend left=15, looseness=1.00] (11.center) to (3);
		\draw [bend left=15, looseness=1.00] (13.center) to (0);
		\draw [bend right=15, looseness=1.00] (1.center) to (0);
		\draw (3) to (0);
		\draw (3) to (2.center);
		\draw (0) to (6.center);
	\end{pgfonlayer}
\end{tikzpicture}}
\endpgfgraphicnamed
\end{equation}
where $\nststile{a^{\prime}_2}{a^{\prime}_1}\in \mathcal{A}, \nststile{p^{\prime}_1}{p^{\prime}_1}\in \mathcal{P}, \nststile{m^{\prime}_2}{m^{\prime}_1}\in \mathcal{M}$.

\end{lemma}

\section{Completeness}

\subsection{Comparing rGS-LC diagrams}
In this subsection, we show that a pair of rGS-LC diagrams can be transformed into such a form that they are equal under the standard interpretation if and only if they  are identical. 

 \begin{definition}\cite{Miriam2}
A pair of rGS-LC diagrams on the same number of qutrit is called simplified if there are no pairs of qutrits $a, b$,  such that $a$ has a red node in its vertex operator in the first diagram but not in the second, $b$ has a red node in the second diagram but not in the first, and $a$ and $b$ are adjacent in at least one of the diagrams.
 \end{definition}

\begin{lemma}\label{rgspsym}
Each pair of rGS-LC diagrams on the same number of qutrits can be made into a simplified form.
\end{lemma}
\begin{proof}
The proof is same as that of the qubit case presented in   \cite{Miriam1, Miriam2}.
 \end{proof}

\begin{lemma}\label{unpaired}
Any component  diagrams  of a simplified pair of rGS-LC with  an unpaired red node are unequal, where this red node  resides as a vertex operator only  in one of the pair of diagrams. 
\end{lemma}


\begin{theorem}\label{identical}
The component  diagrams of any simplified pair of rGS-LC  are equal under the standard interpretation if and only if they  are identical. 

\end{theorem}


\subsection{Completeness for qutrit stabilizer quantum mechanics}
To achieve the proof of completeness for qutrit stabilizer QM, we will proceed in two main steps.   Firstly, we show the completeness for stabilizer states. Then by theorem \ref{gslctorgslc} and theorem \ref{identical} we have
\begin{theorem}\label{complst}
The ZX-calculus is complete for pure qutrit stabilizer states.

\end{theorem}
Then we use the following map-state duality to relate quantum states and linear operators:

\begin{equation}\label{mapstate}
\beginpgfgraphicnamed{Qutrits/mapstatedual}
\begin{tikzpicture}
	\begin{pgfonlayer}{nodelayer}
		\node [style=none] (0) at (-6, 0.25) {};
		\node [style=none] (1) at (-6, 0.25) {};
		\node [style=none] (2) at (-6.5, 1.25) {};
		\node [style=none] (3) at (-5.25, -0.25) {};
		\node [style=gn] (4) at (-6.25, -0.25) {};
		\node [style=gn] (5) at (-6, -1) {};
		\node [style=none] (6) at (-5.25, 0.25) {};
		\node [style=none] (7) at (-6.5, 0.25) {};
		\node [style=none] (8) at (-7, 0) {};
		\node [style=none] (9) at (-7, 1.25) {};
		\node [style=none] (10) at (-5, 0.75) {};
		\node [style=none] (11) at (-5, 0.25) {};
		\node [style=none] (12) at (-6.25, 0.25) {};
		\node [style=none] (13) at (-6.25, 0.75) {};
		\node [style=none] (14) at (-5.75, 0.5) {$A$};
		\node [style=none] (15) at (-6, 0.75) {};
		\node [style=none] (16) at (-5.25, 0.75) {};
		\node [style=none] (17) at (-5.25, 1.25) {};
		\node [style=none] (18) at (-6, 1.25) {};
		\node [style=none] (19) at (-5.75, 0) {$\cdots$};
		\node [style=none] (20) at (-5.75, 1) {$\cdots$};
		\node [style=none] (21) at (-6.75, 1) {$\cdots$};
		\node [style=none] (22) at (-4.5, 0) {$=$};
		\node [style=none] (23) at (-1.5, 0.25) {};
		\node [style=none] (24) at (-3.75, -0.25) {};
		\node [style=none] (25) at (-3.75, 0.25) {};
		\node [style=none] (26) at (-1.5, -0.25) {};
		\node [style=none] (27) at (-2.5, 0) {$B$};
		\node [style=none] (28) at (-3.5, 0.75) {};
		\node [style=none] (29) at (-3, 0.75) {};
		\node [style=none] (30) at (-2.25, 0.75) {};
		\node [style=none] (31) at (-1.75, 0.75) {};
		\node [style=none] (32) at (-3.5, 0.25) {};
		\node [style=none] (33) at (-3, 0.25) {};
		\node [style=none] (34) at (-2.25, 0.25) {};
		\node [style=none] (35) at (-1.75, 0.25) {};
		\node [style=none] (36) at (-3.25, 0.5) {$\cdots$};
		\node [style=none] (37) at (-2, 0.5) {$\cdots$};
		\node [style=none] (38) at (0, 0) {$\Longleftrightarrow$};
		\node [style=none] (39) at (2, -0.25) {$\cdots$};
		\node [style=none] (40) at (2.75, 0.5) {};
		\node [style=none] (41) at (1.75, 1) {};
		\node [style=none] (42) at (1.5, 0) {};
		\node [style=none] (43) at (3.25, 0) {$=$};
		\node [style=none] (44) at (2.5, 0.5) {};
		\node [style=none] (45) at (1.75, 0) {};
		\node [style=none] (46) at (1.75, 0.5) {};
		\node [style=none] (47) at (2, 0.75) {$\cdots$};
		\node [style=none] (48) at (1.5, 0.5) {};
		\node [style=none] (49) at (2.5, 1) {};
		\node [style=none] (50) at (2.75, 0) {};
		\node [style=none] (51) at (2.5, 0) {};
		\node [style=none] (52) at (1.75, 0) {};
		\node [style=none] (53) at (2, 0.25) {$A$};
		\node [style=none] (54) at (1.75, -0.5) {};
		\node [style=none] (55) at (2.5, -0.5) {};
		\node [style=none] (56) at (7, 0) {};
		\node [style=none] (57) at (5.5, 0) {};
		\node [style=none] (58) at (6.25, 0) {};
		\node [style=none] (59) at (7, -0.5) {};
		\node [style=none] (60) at (6.75, 0.5) {};
		\node [style=none] (61) at (6.25, 0.5) {};
		\node [style=none] (62) at (5.25, 0.25) {$\cdots$};
		\node [style=none] (63) at (6, -0.25) {$B$};
		\node [style=none] (64) at (6.5, 0.25) {$\cdots$};
		\node [style=none] (65) at (6.75, 0) {};
		\node [style=none] (66) at (5, 0) {};
		\node [style=none] (67) at (4.75, -0.5) {};
		\node [style=none] (68) at (4.75, 0) {};
		\node [style=gn] (69) at (4.75, 0.75) {};
		\node [style=none] (70) at (5, 0) {};
		\node [style=none] (71) at (4.75, 0.25) {};
		\node [style=none] (72) at (4.5, 0.25) {};
		\node [style=none] (73) at (5, 0.25) {};
		\node [style=none] (74) at (4.5, -0.5) {};
		\node [style=gn] (75) at (4.75, 1.25) {};
		\node [style=none] (76) at (5.5, 0) {};
		\node [style=none] (77) at (4, 0.5) {};
		\node [style=none] (78) at (4, -0.5) {};
		\node [style=none] (79) at (5.5, 0.5) {};
		\node [style=none] (80) at (4.25, -0.25) {$\cdots$};
	\end{pgfonlayer}
	\begin{pgfonlayer}{edgelayer}
		\draw [bend left, looseness=1.10] (5) to (8.center);
		\draw [bend right, looseness=1.10] (5) to (3.center);
		\draw (9.center) to (8.center);
		\draw (3.center) to (6.center);
		\draw [bend right, looseness=1.00] (4) to (1.center);
		\draw [bend left, looseness=0.85] (4) to (7.center);
		\draw (0.center) to (1.center);
		\draw (7.center) to (2.center);
		\draw (13.center) to (12.center);
		\draw (13.center) to (10.center);
		\draw (10.center) to (11.center);
		\draw (12.center) to (11.center);
		\draw (18.center) to (15.center);
		\draw (17.center) to (16.center);
		\draw (25.center) to (24.center);
		\draw (25.center) to (23.center);
		\draw (23.center) to (26.center);
		\draw (24.center) to (26.center);
		\draw (28.center) to (32.center);
		\draw (29.center) to (33.center);
		\draw (30.center) to (34.center);
		\draw (31.center) to (35.center);
		\draw (45.center) to (52.center);
		\draw (48.center) to (42.center);
		\draw (48.center) to (40.center);
		\draw (40.center) to (50.center);
		\draw (42.center) to (50.center);
		\draw (41.center) to (46.center);
		\draw (49.center) to (44.center);
		\draw (52.center) to (54.center);
		\draw (51.center) to (55.center);
		\draw (68.center) to (67.center);
		\draw (68.center) to (56.center);
		\draw (56.center) to (59.center);
		\draw (67.center) to (59.center);
		\draw (61.center) to (58.center);
		\draw (60.center) to (65.center);
		\draw [bend left, looseness=1.00] (69) to (73.center);
		\draw [bend right, looseness=0.85] (69) to (72.center);
		\draw (70.center) to (73.center);
		\draw (74.center) to (72.center);
		\draw [bend right, looseness=1.00] (75) to (77.center);
		\draw [bend left, looseness=1.00] (75) to (79.center);
		\draw (78.center) to (77.center);
		\draw (79.center) to (76.center);
	\end{pgfonlayer}
\end{tikzpicture}}
\endpgfgraphicnamed
\end{equation}


Finally we have the main result: 
\begin{theorem}\label{complall}
The ZX-calculus is complete for qutrit stabilizer quantum mechanics.
\end{theorem}

\section{Conclusion and further work}

In this paper, we  show that  the qutrit ZX-calculus is complete for pure qutrit stabilizer quantum mechanics, using the same strategy as in the qubit case. 

An obvious next step is to extend the results of qubit ZX-calculus to qutrit ZX-calculus: i.e. give a counter-example to show the incompleteness of ZX-calculus for overall qutrit  quantum mechanics, which can not be obtained by trivially following the results in \cite{Vladimir} because of a lack of Euler decomposition of an arbitrary $3\times 3$ unitary into $Z$ and $X$ phases.  In view of the very recent proof of  completeness of  ZX-calculus for qubit Clifford $+T$ quantum mechanics \cite{Emmanuel}, we hope to achieve the completeness result  in the qutrit case,  where the qutrit version of $T$  gate has  been defined in  \cite{amk} and \cite{Shawn}. Moreover, we would like to generalise the main result of this paper to qudit ZX-calculus for arbitrary dimension $d$.

Another question suggested by Bob Coecke is to embed the qubit ZX-calculus into  qutrit ZX-calculus, which is not obvious since a representation of some special non-stabilizer  phase is involved.
It would also be interesting to incorporate the rules of qutrit ZX-calculus in the automated graph rewriting system Quantomatic  \cite{Quanto} to help for searching for new quantum error-correcting codes  \cite{ckzh}.

 \section*{Acknowledgement}
 The author would like to thank Bob Coecke for the useful discussions and comments. The suggestions and comments given by the anonymous reviewers are also greatly appreciated.
\bibliographystyle{eptcs}
\bibliography{refs}

\end{document}